\documentclass[11pt,a4paper,reqno]{amsart}
\usepackage{fullpage} 
\usepackage{amsfonts}
\usepackage{amsthm}
\usepackage[latin9]{inputenc}
\usepackage[T1]{fontenc}
\usepackage{amsmath} 
\usepackage{graphicx}
\usepackage{enumitem}
\usepackage{wrapfig}
\usepackage{dsfont}
\usepackage{braket}
\usepackage{comment}
\usepackage{amsxtra}
\usepackage{stackengine}
\usepackage{scalerel}
\usepackage{amssymb}
\usepackage{amsmath,mathtools}
\usepackage{hyperref}
\allowdisplaybreaks[1]

\newtheoremstyle{break}
  {}
  {}
  {\itshape}
  {}
  {\bfseries}
  {.}
  {\newline}
  {}

\theoremstyle{break}

\newtheorem{de}{Definition}
\newtheorem{theo}[de]{Theorem}    
\newtheorem{prop}[de]{Proposition}
\newtheorem{lem}[de]{Lemma}
\newtheorem{rem}[de]{Remark}

\DeclareMathOperator{\tr}{Tr}

\DeclareMathOperator{\supp}{supp}

\DeclareMathOperator{\loc}{loc}

\DeclareMathOperator{\ima}{Im}
\DeclareMathOperator{\spn}{span}

\title{Derivation of the dipolar Gross-Pitaevskii energy}
\date{}
\author{Arnaud Triay}
\address{CEREMADE, CNRS, Université Paris-Dauphine, PSL Research University, 75016 Paris, France}
\email{triay@ceremade.dauphine.fr}
\begin{document}
\maketitle
\begin{abstract}
We consider $N$ trapped bosons in $\mathbb{R}^3$ interacting via a pair potential $w$ which has a long range of dipolar type. We show the convergence of the energy and of the minimizers for the many-body problem towards those of the dipolar Gross-Pitaevskii functional, when $N$ tends to infinity. In addition to the usual cubic interaction term, the latter has the long range dipolar interaction. Our results hold under the assumption that the two-particle interaction is scaled in the form $N^{3\beta - 1} w(N^\beta x)$ for some $0 \leq \beta < \beta_{\max}$ with $\beta_{\max} = 1/3 + s/(45+42s)$ where $s$ is related to the growth of the trapping potential.
\end{abstract}

\tableofcontents

\section{Introduction}
Bose-Einstein Condensation (BEC) is a phenomenon occurring at very low temperature for a highly dilute gas of bosons. In the proper experimental conditions, most of the particles get to occupy the same quantum state. In 1925, Bose~\cite{Bose-24} and Einstein~\cite{Ein-24} proved condensation for ideal particles, that is, under the important assumption that the particles do not interact with each other. Seventy years later, Cornell and Wieman~\cite{CorWie-95} obtained the first experimental realization of a full BEC with Rubidium atoms. This major achievement has triggered a new interest in the theoretical study of condensation. In 1999 the first mathematical proof of BEC was provided by Lieb, Seiringer and Yngvason~\cite{LieSeiYng-00,LieSeiSolYng-05} for confined particles interacting via positive, radial and short range pair interaction, in the dilute regime.

The experimental study of BEC remains a very active field of research today. A challenging task is to realize condensates of particles with diverse interactions. In 2005, a first dipolar condensate has been observed by Griesmaier \emph{et al} using Chromium particles \cite{GriWerHenStuPfa-05}. The dipolar interaction differs in many points from the ones encountered before. It is long range, anisotropic and has an attractive and a repulsive part. Its study is both experimentally and mathematically intricate but it opens the way to new physical effects. For instance, it is believed that dipolar condensates exhibit a roton-maxon excitation spectrum \cite{SanShlLew-03} which is not observed for simpler interactions. See for example \cite{LahMenSanLewPfau-09} for a review of the physical properties of the dipolar Bose gas.

In this article, we give the first derivation of the Gross-Pitaevskii (GP) theory for dipolar Bose gases, starting from the many-particle linear Schrödinger problem. Our method is based on mean field limits and de Finetti theorems introduced in \cite{LewNamRou-13,LewNamRou-15} as well as new techniques developed to deal with negative pair interaction that can been found in \cite{NamRouSer-15,LewNamRou-15b}.

More precisely, we will show that the particles in the condensate have a common state which can be computed by minimizing the dipolar Gross-Pitaevskii functional:
\begin{multline}\label{GP_dip_functional}
\mathcal{E}_{GP}^{a,b}(u) = \int_{\mathbb{R}^{3}} |(\nabla + i A(x))u(x)|^2 dx + \int_{\mathbb{R}^{3}} V(x) |u(x)|^2 dx + \frac{a}{2} \int_{\mathbb{R}^{3}} |u(x)|^4dx\\ + \frac{b}{2} \int_{\mathbb{R}^{3}\times\mathbb{R}^{3}} (K \star |u|^2(x)) |u(x)|^2dx.
\end{multline}
The first term in (\ref{GP_dip_functional}) represents the kinetic energy where $A: \mathbb{R}^{3} \to \mathbb{R}^{3}$ is a vector potential (modeling a magnetic field or the Coriolis force due to the rotation of the atoms). The second term is the one-body potential energy where $V:\mathbb{R}^{3} \to \mathbb{R}$ is a trapping potential, that is, $V(x) \to \infty$ when $|x|\to\infty$. The third term is the short range interaction in the gas where $a\in \mathbb{R}$ is proportional to the scattering length or an approximation of it. The last term is the dipolar energy and $b\in\mathbb{R}$ is proportional to the norm of the dipoles. The dipolar interaction potential
\begin{equation}\label{dip_pot}
K_{dip}(x) = \frac{1-3\cos^2(\theta_x)}{|x|^3} = \frac{\Omega_{dip}(x/|x|)}{|x|^3},
\end{equation}
represents the interaction between two aligned dipoles located at distance $|x|$. The parameter $\theta_x$ is the angle between $x$ and the direction $n$ of all the dipoles, namely $\cos(\theta_x) = n \cdot x / |x|$. In this paper we consider more general long-range potentials of the form
\begin{equation}\label{K_general}
K(x) = \frac{\Omega(x/|x|)}{|x|^3}
\end{equation}
where $\Omega$ is a even function satisfying the cancellation property on $\mathbb{S}^{2}$, the unit sphere of $\mathbb{R}^{3}$,
\begin{equation}\label{cancellation_prop_0}
\int_{\mathbb{S}^2} \Omega(\omega) d\sigma(\omega) = 0,
\end{equation}
where $d\sigma$ denotes the Lebesgue measure on $\mathbb{S}^{2}$. Note that the convolution
$$K\star |u|^2(x) = \int_{\mathbb{R}^{3}} \frac{\Omega(y/|y|)}{|y|^3} |u(x-y)|^2 dy = \lim_{\varepsilon\to 0}  \int_{|y|>\varepsilon} \frac{\Omega(y/|y|)}{|y|^3} |u(x-y)|^2 dy$$
 is an improper integral which is not absolutely convergent. It will be discussed later in Lemma \ref{lem_dip_cont}.
 If all the particles share the same quantum state $u$, the latter should minimize the functional in (\ref{GP_dip_functional}) under the normalization constraint
\begin{equation*}
\int_{\mathbb{R}^{3}} |u(x)|^2dx = 1.
\end{equation*}
We therefore introduce the ground state energy 
\begin{equation}\label{GP_min}
\boxed{
e_{GP}(a,b) := \inf_{\substack{\|u\|_{L^2}=1}} \mathcal{E}_{GP}^{a,b}(u).
}
\end{equation}
This variational problem has been extensively studied, both theoretically \cite{CarMarkSpa-08,CarHaj-14,AntSpa-11} and numerically \cite{BaoAbdCai,BaoCaiWan-10} in the dipolar case (\ref{dip_pot}).

Our aim is to justify the validity of the Gross-Pitaevskii minimization (\ref{GP_min}), starting with the exact many-body problem based on the Hamiltonian
\begin{equation}\label{H_N}
	H_N = \sum_{j=1}^N \bigg( -\big(\nabla_{x_j} +iA(x_j)\big)^2 + V(x_j)\bigg) + \frac{1}{N-1} \sum_{1\leq i < j \leq N} N^{3\beta} w( N^\beta (x_i-x_j)).
\end{equation}
This operator acts on $\bigotimes^N_s L^{2}(\mathbb{R}^{3})$, the symmetric tensor product of $N$ copies of $L^{2}(\mathbb{R}^{3})$. We denote by
\begin{equation*}
\boxed{
e_N := \inf_{\substack{\Psi \in \bigotimes^N_s L^{2}(\mathbb{R}^{3}) \\ \|\Psi\|_{L^{2}}=1}} \frac{\braket{\Psi, H_N \Psi}}{N}
}
\end{equation*}
the many-body ground state energy per particle. We investigate the limit of a large number of particles, $N\to\infty$. In (\ref{H_N}), the scaling chosen for the interaction between the particles is very common. On the one hand, the $L^{1}$--preserving scaling of the potential $$w_N (x) := N^{3\beta} w (N^\beta x)$$ will retain in the limit only its short and long range parts. In our case, we will have
$$
w_N \rightharpoonup a \delta + b K,
$$
since $K$ and $\delta$ behave the same under scaling. On the other hand, the coupling factor $(N-1)^{-1}$ is typical of mean field limits and ensures that the potential is of the same order as the kinetic energy. The function $w$ is the real interaction between the atoms in the gas. It will be assumed to be repulsive at short distances (hence $a>0$) and of dipolar type (close to $b K$) at large distances.

The parameter $\beta \in [0,1]$ interpolates between the pure mean-field regime $\beta = 0$ and the Gross-Pitaevskii regime $\beta = 1$ which is more difficult to handle. The case $\beta<1/3$ corresponds to a high density regime with the interaction length $N^{-\beta}$ being larger than the mean distance $N^{-1/3}$ between the particles. This is a natural setting for the law of large numbers to apply. The case $\beta >1/3$ is more subtle and corresponds to a low density regime where the particles meet rarely but interact with intensity proportional to $N^{3\beta -1}$. The difficulties culminate at $\beta = 1$ where now the details of the scattering process play a role. When $w$ has a negative part, it is natural to assume that the true interaction $w$ is (classically) stable of the second kind, which means that 
\begin{equation*}
\sum_{1\leq i < j \leq N} w(x_i-x_j) \geq - C N,
\end{equation*}
for all $N\geq 2$ and all space configurations $x_1,...,x_N \in \mathbb{R}^{3}$ of the $N$ particles \cite{Ruelle}. This assumption ensures automatically that $e_N$ is bounded from below when $\beta \leq 1/3$. The case $\beta >1/3$ is particularly difficult and the proof that $e_N$ is bounded from below requires to use the kinetic energy, that is, the quantum feature of the system. The scaled interaction alone is unstable.

For $\beta = 0$ the convergence of $e_N$ to the Hartree energy was proven in \cite{LewNamRou-13}. In this paper, we prove that 
\begin{equation}\label{lim_e_N}
\boxed{
\lim_{N\to \infty}e_N = e_{GP}(a,b)}
\end{equation}
as $N\to\infty$ for $$\boxed{0< \beta < \frac{1}{3} + \frac{s}{45+42 s}}$$ under the sole assumption that $w$ is stable as in (\ref{classical_stab}) and that $V(x) \geq C |x|^s$ at infinity. In that case, we are able to show that $e_N \geq - C$ but it is unclear if additional assumptions on $w$ are needed for higher $\beta$'s. Under the additional condition that $e_N \geq -C$ our result (\ref{lim_e_N}) applies up to $\beta < 2/3$. In addition to the convergence of the energy, we are also able to prove the convergence of the minimizers. Loosely speaking, we prove that the ground state of (\ref{GP_min}) factorizes in the limit, that is,
\begin{equation}\label{factorization}
\Psi(x_1,...,x_N) \simeq u(x_1)\cdots u(x_N) \text{ as } N \to \infty,
\end{equation}
where $u$ is a minimizer of the dipolar Gross-Pitaevskii functional. The meaning of (\ref{factorization}) is in the sense of density matrices and not in $L^{2}$ norm as we will recall.

The paper is organized as follows. In Section \ref{main_results} we state our main contributions, including some results on the existence and uniqueness of minimizers of the Gross-Pitaevskii functional (\ref{GP_dip_functional}) as well as our main convergence property (\ref{lim_e_N}). In Section \ref{Big_Proof} we prove our main result, Theorem~\ref{theo_1}.


\section{Main results}\label{main_results}

\subsection{Properties of the dipolar Gross-Pitaevskii functional}
As discussed before, the dipolar GP equation has already been widely studied and we slightly extend the results of \cite{BaoCaiWan-10,CarHaj-14} to deal with more general $K$ and $V$ as well as a magnetic potential or a Coriolis force $A$. The main adjustment concerns the uniqueness of minimizers which, in order to remain true, requires some extra assumptions on the intensity of the magnetic field.
\begin{theo}[Existence of GP minimizers]\label{theo_GP}
Let us assume that $V\in L^1_{\loc}(\mathbb{R}^{3})$, $A\in L_{\loc}^{2}(\mathbb{R}^{3})$ and that there is some $s>0$ such that
\begin{equation}\label{A_and_V_condition_1}
V(x) \geq C^{-1}(|A(x)|^2 + |x|^s) - C,\quad \forall x \in \mathbb{R}^{3}
\end{equation}
for some $C>0$. Let $$K(x) = \frac{\Omega(x/|x|)}{|x|^3}$$ with $\Omega\in L^q(\mathbb{S})$, for some $q>1$, a even function satisfying the cancellation property $$\int_{\mathbb{S}^2} \Omega(w) d\sigma(w) = 0.$$ Let $\widehat{K}$ be the Fourier transform of $K$ in the sense of the principal value (as defined later in Lemma \ref{lem_dip_cont}). \\
i)~ If
\begin{equation}\label{hyp1_theo_GP}
\left \{
\begin{array}{c c}
    b > 0 \text{ and } a \geq b \;(\inf \widehat{K})_-,\\
    \text{ or } \\
    b < 0 \text{ and } a \geq -b \; (\sup \widehat{K})_+, \\
\end{array}
\right.
\end{equation}
then $e_{GP}(a,b) > -\infty $ and $\mathcal{E}_{GP}$ has minimizers.\\
ii)~If
\begin{equation}\label{hyp2_theo_GP}
\left \{
\begin{array}{c c}
    b > 0 \text{ and } a < b \;(\inf \widehat{K})_-,\\
    \text{ or } \\
    b < 0 \text{ and } a < -b \; (\sup \widehat{K})_+, \\
\end{array}
\right.
\end{equation}
then $e_{GP}(a,b) = -\infty$.
\end{theo}

\begin{rem}
In the dipolar case where $\Omega(x/|x|) = 1 - 3 \cos^2(\theta_x)$,  then $\widehat{K}(k) = \frac{4\pi}{3}\left(3\cos^2(\theta_k) - 1 \right)$, $(\inf \widehat{K})_- = 4\pi /3$ and $(\inf \widehat{K})_+ = 8\pi/3$.
\end{rem}

The understanding of the effective theory is necessary to grasp the forthcoming difficulties in the derivation from the many-body theory. The condition \textit{i)} on $a$ and $b$ aims to make the interaction stable which in this case is equivalent to be positive, mainly because then $a + b \widehat{K} \geq 0$. Indeed, one can easily see that if the interaction term is negative for some configuration $(a,b)$ and some wave function $u$, then by scaling one gets that $e_{GP}(a,b) = -\infty$.

In fact, most of our results remains true if we change $(\nabla + i A(x))^2$ by any abstract self-adjoint operator $h$ such that $ h \geq C^{-1}(-\Delta + |x|^s) - C$.

We first state some lemma that will be useful in the proof of Theorem \ref{theo_GP}.

\begin{lem}[Magnetic Laplacian]\label{lem_magn_lapl}
Under assumption (\ref{A_and_V_condition_1}), one has the operator inequality on $L^{2}(\mathbb{R}^{3})$
\begin{equation*}
C^{-1}(-\Delta + V) - C \leq (-i\nabla + A)^2 + V \leq C (-\Delta + V + 1) ,
\end{equation*}
for some constant $C>0$.
\end{lem}
\begin{proof}[Proof of Lemma \ref{lem_magn_lapl}]
Using the Cauchy-Schwarz inequality for operators (\ref{CS_op}) and inequality (\ref{A_and_V_condition_1}), we have the upper bound
\begin{align*}
(p+A)^2 &= p^2 + pA + Ap + |A|^2 \leq C(p^2 + V) + C.
\end{align*}
To get the lower bound, we choose some $0 < \eta < 1$ sufficiently close to $1$ but fixed and we use again the Cauchy-Schwarz inequality
\begin{align*}
(p+A)^2 +V &\geq (1-\eta) p^2 + (1-\eta^{-1}) |A|^2 + V \geq C^{-1}(p^2+V) - C.
\end{align*}
\end{proof}

\begin{proof}[Proof of Theorem \ref{theo_GP}]
We start by proving \textit{i)}. First, since $V\in L^{1}_{\loc}(\mathbb{R}^{3})$ and $A\in L^{2}_{\loc}(\mathbb{R}^{3})$ the quadratic form associated with the first two terms of the energy $\mathcal{E}_{GP}^{a,b}$ is closed on a domain included in $H^1(\mathbb{R}^{3})$ and provides a self-adjoint realization of $H_0 = -(\nabla + iA)^2 + V$ by the method of Friedrichs \cite{LeiFel-81}. Then, to prove \textit{i)}, it suffices to follow the same proof as in \cite{CarHaj-14} and to use, when necessary, that $H_0 = -(\nabla + iA)^2 + V$ has compact resolvent \cite{Iwatsuka-86}. More precisely, let us take a minimizing sequence $(u_n)$ for $\mathcal{E}_{GP}$. The condition on $a$ and $b$ ensures that the interaction part of the energy is non-negative $$a \int |u|^4 + b \int K \star |u|^2 |u|^2 = \int \left( a + b \widehat{K}\right) |\widehat{|u|^2}|^2 \geq 0.$$ Thus $\mathcal{E}_{GP}$ is bounded below on the unit sphere of $L^2(\mathbb{R}^{3})$ and $\braket{u_n, H_0 u_n}$ is bounded. We can then extract of $(u_n)$ a converging subsequence in $L^2 \cap L^4$. Indeed, thanks to Lemma \ref{lem_magn_lapl}, $H_0$ has compact resolvent. We can write $u_n = H^{-1/2}_0 H_0^{1/2}u_n$ and since $(H_0^{1/2} u_n)$ is bounded in $L^2(\mathbb{R}^{3})$, because $(u_n)$ is a minimizing sequence of $\mathcal{E}_{GP}^{a,b}$, then $(u_n)$ is precompact in $L^2(\mathbb{R}^{3})$. We can then extract a converging subsequence that we still denote by $(u_n)$ and we write $u$ its limit. Fatou's lemma gives $\mathcal{E}(u) \leq \liminf \mathcal{E}_{GP}(u_n) = \inf \mathcal{E}_{GP}$.

 For \textit{ii)}, it suffices to employ Lemma \ref{lem_magn_lapl} above to adapt the proof of \cite{CarHaj-14}. Without loss of generality, we can assume that $e_3 \in \{a + b\widehat{K} < -\delta \} \neq \emptyset$. Let $f$ be a smooth function compactly supported with $\|f\|_{L^2} = 1$ and let us denote by $\rho (x) = |f|^2 (x)$ and $f_\lambda(x) = f(\lambda x_1 , \lambda x_2, \lambda^{-2}x_3)$. We have $\|f_\lambda\|_{L^2} = \|f\|_{L^2} = 1$, $\| \widehat{\rho}_\lambda^2 \|_{L^1} = \| \widehat{\rho}^2\|_{L^1} = \|f\|_{L^4}^4$ and $\|\widehat{\rho}_\lambda\|_{L^\infty} \leq \|\rho_\lambda\|_{L^1} = 1$ for all $\lambda >0$. For any $r>0$ we denote by $C_r$ the cone with vertex at the origin, with direction $e_3$ and with angle $r$. For any $\eta > 0$, there exists $\lambda_0$ such that for any $0<\lambda \leq \lambda_0$, we have 
 \begin{equation}\label{tendue}
\int_{\mathbb{R}^{3}\setminus C_r} \rho^2_\lambda < \eta.
 \end{equation}
According to Lemma \ref{lem_dip_cont} below, $\widehat{K}$ is an homogeneous function which is continuous except at the origin. Since $\widehat{K}$ is continuous at $e_3$, we can find $r>0$ such that $B(e_3,r) \subset \{ a + b\widehat{K} < -\delta\}$. Now for any $\eta >0$, there exists some $\lambda_0>0$ such that for all $\lambda \leq \lambda_0$ we have (\ref{tendue}), then
 \begin{align*}
 \int \left(a+b\widehat{K}\right) |\widehat{\rho}_\lambda|^2
 &=  \int_{\mathbb{R}^{3}\setminus C_r} \left(a+b\widehat{K}\right) |\widehat{\rho}_\lambda|^2 + \int_{C_r} \left(a+b\widehat{K}\right) |\widehat{\rho}_\lambda|^2 \\
 &\leq \eta \|a + b\widehat{K}\|_{L^\infty} -(\|f\|_{L^4}^4-\eta) \delta.
 \end{align*}
If we take $\eta < \delta \|f\|_{L^4}^4 / ( 2\delta + 2\|a + b\widehat{K}\|_{L^\infty})$, we obtain
\begin{equation*}
\int a |f_\lambda|^4 + b K\star |f_\lambda|^2 |f_\lambda|^2 = \int \left(a+b\widehat{K}\right) \widehat{\rho}_\lambda^2 \leq -\frac{\|f\|_{L^4}^4 \delta}{2} < 0.
\end{equation*}
Now, let us denote by $\varphi = f_{\lambda_0}$ and by $\varphi_\ell = \ell^{3/2} \varphi (\ell x)$. Since $V$ is $L^1_{loc}$ and $\varphi$ is compactly supported, with, say, $\supp \varphi \subset B(0,R)$, we have
$$
\left| \int_{\mathbb{R}^{3}} V |\varphi_\ell|^2 \right| \leq \ell^3 \int_{B(0,R\ell^{-1})} |V| = o(\ell^{3}).
$$
Applying Lemma \ref{lem_magn_lapl} we obtain
\begin{align*}
\mathcal{E}_{GP}(\varphi_\ell) \leq C \left( \int |\nabla \varphi_\ell|^2 + V|\varphi_\ell|^2\right) + \int a |\varphi_\ell|^4 + b K\star |\varphi_\ell|^2 |\varphi_\ell|^2.
\end{align*}
And for $\ell$ large enough, we have
\begin{align*}
e_{GP}(a,b) &\leq C \ell^2 \int |\nabla \varphi|^2 +\ell^3 \int_{B(0,R\ell^{-1})} |V| + \ell^3 \left(\int a |\varphi|^4 + b K\star |\varphi|^2 |\varphi|^2\right) \\
&\leq C \ell^2 - \frac{\|f\|_{L^4}^4 \delta}{2} \ell^3.
\end{align*}
Taking the limit $\ell\to\infty$, the last inequality gives $e_{GP}(a,b) = - \infty$.
\end{proof}

If $A=0$, we can assume $u > 0$, because $|\nabla u| \geq |\nabla |u||$ almost everywhere. Then since $a+b\widehat{K}\geq 0$, the functional is strictly convex with respect to $|u|^2$ and uniqueness follows as in \cite{BaoCaiWan-10} where the case $K = K_{dip}$ was considered. On the other hand, if $A\neq 0$, the functional is no more convex and uniqueness can fail. Indeed, a  strong magnetic field can create vortices which are a sign of rotational symmetry breaking. This situation was already encountered in the non dipolar case $K=0$ in \cite{Sei-02,Sei-03}. A small magnetic field $A$ does not create any vortex and uniqueness remains true, as we will show in Theorem \ref{theo_uniq}. 

In the non-dipolar case $K=0$, uniqueness has been shown in several situations \cite{LieSeiSolYng-05,Sei-02,Sei-03}. These works can all be adapted to the dipolar case $K\neq 0$ but here for shortness we only discuss an extension of \cite[Ch. 7]{LieSeiSolYng-05}. In the following result we apply the implicit function theorem for $A$ small enough without getting any information on how small $A$ has to be. The second part deal with a radial magnetic field and gives an explicit range of validity as in \cite[Ch. 7]{LieSeiSolYng-05}.
\begin{theo}[Uniqueness of the minimizer for the dipolar GP functional with magnetic field]\label{theo_uniq}
Let $A, V$ and $K$ satisfy the assumptions of Theorem \ref{theo_GP} and let $(a,b)$ be admissible parameters as in (\ref{hyp1_theo_GP}).
\begin{enumerate}
\item \label{lem_uniq_0} For $t \in \mathbb{R}$, we denote by $\mathcal{E}_{GP}^{a,b,t}$ the Gross-Pitaevskii functional where $A$ has been replaced by $t A$. Then, there exists $t_0 > 0$ such that for all $|t| < t_0$, $\mathcal{E}_{GP}^{a,b,t}$ has a unique minimizer, up to a phase factor, in the sector of mass $\int_{\mathbb{R}^{3}} |u|^2 = 1$.

\item \label{lem_uniq_1} Assume $A(x) = a(r,z)e_\theta$, $V(x) = V(r,z)$ where $(r,\theta,z)$ are the cylindrical coordinates with $e_z = n$ (the commmon direction of all the dipolar moments). If $\|r a\|_{\infty} < 1/2$ then $\mathcal{E}_{GP}$ has a unique minimizer, up to a phase factor, in the sector of mass $\int_{\mathbb{R}^{3}} |u|^2 =1$. It is non negative and axially symmetric, with axis $e_z=n$.
\end{enumerate}
\end{theo}

The proof is given in Appendix \ref{section_proof_lem_uniq_0}.

\subsection{Derivation of the dipolar Gross-Pitaevskii energy and minimizers}\label{section_deriv_GP}
The phenomenon of condensation can be detected through the convergence of the ground state and the ground state energy towards those predicted by the Gross-Pitaevskii theory. We want to give general assumptions on $w$ to be stable and to behave like the dipolar interaction at long distances. Investigating this problem, one can first notice that taking $w= w_0 + K$ with $w_0\in L^1$ is not permitted. Indeed, the two particles hamiltonian $H_2$ is ill defined due to the singularity of $K$ at the origin. The dipolar interaction is only physically relevant at large distances and therefore it is reasonable to assume that $w$ is close to $K$ outside of a ball.

\begin{theo}[Convergence of the energy]\label{theo_1}
Let $A\in L^{2}_{loc}(\mathbb{R}^{3})$ and $0 \leq V\in L^1_{loc}(\mathbb{R}^{3})$ satisfying
\begin{equation*}
V(x) \geq C^{-1}(|A(x)|^2 + |x|^s) - C, \quad \forall x \in \mathbb{R}^{3}.
\end{equation*}
Let $K = \Omega(x/|x|)|x|^{-3}$ with $\Omega\in L^q(\mathbb{S}^{2})$, for some $q\geq 2$, a even function satisfying the cancellation property (\ref{cancellation_prop_0}). Let $w$, $R >0$ and $b\in \mathbb{R}$, be such that 
\begin{equation*}
w-b\mathds{1}_{|x|>R}K \in L^1(\mathbb{R}^{3})\cap L^2(\mathbb{R}^{3}).
\end{equation*}
We also assume $w$ to be classically stable, that is
\begin{equation}\label{classical_stab}
\sum_{1\leq i < j \leq N} w(x_i-x_j) \geq - C N, \quad \forall N \geq 2, \quad \forall x_1,...,x_N \in \mathbb{R}^{3},
\end{equation}
for some $C>0$ independent of $N$. \\
If 
\begin{equation}\label{beta_condition}
\beta < \frac{1}{3} + \frac{s}{45+42s}
\end{equation}
then
\begin{equation}\label{cv_energy}
\boxed{
\lim_{N\to\infty} e_N = e_{GP}(a,b),
}
\end{equation}
for 
\begin{equation*}
a = \int_{\mathbb{R}^{3}} \left(w(x) - b\mathds{1}_{|x|> R}K(x) \right) \;dx \label{def_a_theo}.
\end{equation*} 
In addition, $a$ and $b$ satisfy the stability property (\ref{hyp1_theo_GP}).
\end{theo}
\begin{rem}
The convergence rate given by the proof is the sum of the error in (\ref{bound_on_CV}) and of another one depending on $w$ (see Lemma \ref{lem_from_H_to_nls}) which comes from the approximation of the Gross-Pitaevskii energy by the Hartree one. We do not state any quantitative estimate here for shortness.
\end{rem}
Notice that the value of the constant $R>0$ in (\ref{def_a_theo}) has no importance because of the cancellation property of $K$, see Lemma \ref{lem_dip_cont} below.

If $V(x)\to\infty$ when $|x|\to\infty$ faster than any polynomial (e.g. $V=+\infty$ outside a bounded domain) then the condition (\ref{beta_condition}) reduces to
\begin{equation*}
\beta < \frac{1}{3} + \frac{1}{42}.
\end{equation*}

We now turn to the convergence of states, which is expressed in terms of the $k$-particle reduced density matrices. For $\Psi \in \bigotimes_s^NL^2(\mathbb{R}^{3})$, let us define its $k$-particle density matrix by
\begin{equation*}
\gamma^{(k)}_{\Psi} := \tr_{k+1 \to N} \ket{\Psi}\bra{\Psi},
\end{equation*}
or, in terms of its kernel,
\begin{equation*}
\gamma^{(k)}_{\Psi}(x_1,...,x_k,y_1,...,y_k) = \int_{(\mathbb{R}^{3})^{N-k}} \Psi(x_1,...,x_k,z_{k+1},...,z_N) \overline{\Psi(y_1,...,y_k,z_{k+1},...,z_N)} dz_{k+1}...dz_N.
\end{equation*}
From now on, when we consider a ground state $\Psi_N$ of $H_N$, we will denote by $\gamma^{(k)}_N$ its $k$-particle reduced density matrix for simplicity.

\begin{theo}[Convergence of states]\label{theo_2}
Under the assumptions of Theorem \ref{theo_1}, for any sequence of ground states $(\Psi_N)$, there exists a Borel probability measure $\mu$ supported on $\mathcal{M}_{GP}(a,b)$, the set of ground states of $\mathcal{E}_{GP}^{a,b}$, such that, up to a subsequence $(N')$,
\begin{equation}\label{conv_state}
\gamma^{(k)}_{N'}  \underset{N'\to\infty}{\longrightarrow} \int_{\mathcal{M}_{GP}} \ket{u^{\otimes k}}\bra{u^{\otimes{k}}} d\mu(u), \, \quad \forall k\geq 1,
\end{equation}
where the convergence is in the trace norm. Besides, if $\mathcal{E}_{GP}^{a,b}$ has a unique minimizer, then there is convergence of the whole sequence $(\gamma^{(k)}_N)_N$ in (\ref{conv_state}).
\end{theo}
\begin{proof}
Since the energy per particle is bounded by Theorem \ref{theo_1}, the convergence of states can be proved exactly as in \cite[Theorem 2.5]{LewNamRou-15}.
\end{proof}

We now make several remarks concerning our two main results, Theorem \ref{theo_1} and \ref{theo_2}.\\

The main novelty of these two results is that the interaction potential has a negative part, is anisotropic and long range. Furthermore the derivation holds for some $\beta > 1/3$ where the stability is of quantum nature (due to the kinetic term). Unfortunately, we are not yet able to push the analysis to larger $\beta$'s, but we conjecture that similar results hold for $\beta < 1$, possibly under more stringent assumptions on $w$.

In fact, the assumption (\ref{classical_stab}) is not essential.What we need is that the kinetic energy per particle is bounded, or more precisely,
\begin{equation}\label{hyp_moment_est}
\tr\left(h\otimes h \gamma^{(2)}_{\Psi_N}\right)\leq C
\end{equation}
for all $N$, where $h = (-i\nabla + A)^2 + V$. For example, if $H_{N,\varepsilon} \geq - CN$, where $H_{N,\varepsilon}$ is defined as $H_N$ but replacing $w$ by $(1-\varepsilon)^{-1}w$ in (\ref{H_N}), for some $\varepsilon >0$, we can prove (\ref{hyp_moment_est}) for any $\beta < 2/3$ and therefore have convergence of both the ground states and ground state energies. If, on the other hand, we assume (\ref{classical_stab}), the boundedness of the moment (\ref{hyp_moment_est}) holds immediately for $\beta \leq 1/3$ and with a bootstrap argument borrowed from \cite{LewNamRou-15b} we extend it to $\beta < 1/3 + s/(45+42s)$.

The assumption (\ref{classical_stab}) is very natural. If a potential $w$ satisfies $w-b\mathds{1}_{|x|>R}K_{dip} \in L^1\cap L^2$ then by increasing its value in a neighborhood of the origin, we can make it classically stable as in (\ref{classical_stab}). We explain this in Appendix \ref{section_proof_lem_stab_dip}.

A natural question arises: are all admissible configurations $(a,b)$ (i.e. those satisfying (\ref{hyp1_theo_GP}) of Theorem \ref{theo_GP}) reachable by the derivation? In the case of the dipolar potential
\begin{equation*}
K_{dip}(x) = \frac{1-3\cos^2(\theta_x)}{|x|^3},
\end{equation*}
the answer is yes. We are able to construct a potential $w_{dip}$ (resp. $\widetilde{w_{dip}}$) satisfying the assumptions of Theorem \ref{theo_1} in the case $b>0$ (resp. $b<0$) and such that the corresponding $a$ and $b$ satisfy the equality case $a = \frac{4\pi}{3} b$ (resp. $a = -\frac{8\pi}{3} b$). Then, any configuration $(a,b)$ satisfying the strict inequality (\ref{hyp1_theo_GP}) is reached by adding a well chosen non negative function to $w_{dip}$ (resp. $\widetilde{w_{dip}}$). The interaction $w_{dip}$ is defined as follows. Let $d \in \mathbb{R}$,
\begin{equation}\label{w_dip}
w_{dip}(x) := 2W(x) - W(x+d n) - W(x-d n)
\end{equation}
where
\begin{equation*}
W(x) = \frac{1-e^{-|x|}}{|x|}.
\end{equation*}
The potential $w_{dip}$ represents the interaction of a couple of dipoles interacting with the smeared Coulomb potential $W$. The potential $\widetilde{w_{dip}}$ is defined in Appendix \ref{section_proof_prop_1}, in the proof of the following proposition. One could think that the case $b<0$ has less physical meaning as it will turn out that $b = d^2$, but it is experimentally feasible to tune the parameter $b$ to be negative \cite{GioGorPfa-02}.
\begin{prop}[Full range of parameters in the dipolar GP functional]\label{prop_a_b}
For any admissible parameters $(a,b)$, i.e. satisfying (\ref{hyp1_theo_GP}), there exists a potential $w$ satisfying the assumptions of Theorem \ref{theo_1} and in particular $a = \int (w - b\mathds{1}_{|x|> R}K_{dip})$. Hence we have
\begin{equation*}
\lim_{N\to\infty} e_N = e_{GP}(a,b).
\end{equation*}

\end{prop}

The proof is provided in Appendix \ref{section_proof_prop_1}. The rest of the paper is dedicated to the proof of our main result.

\section{Proof of Theorem \ref{theo_1}: derivation of the dipolar Gross-Pitaevskii energy}\label{Big_Proof}
This section is dedicated to the proof of Theorem \ref{theo_1}.
\subsection{Preliminaries}

Our method in this paper is to follow the path exposed in \cite{LewNamRou-15} which consists in a two steps argument. We first approximate the $N$-body theory by Hartree's theory and then pass from the latter to the Gross-Pitaevskii theory. We use and adapt techniques developed in \cite{NamRouSer-15,LewNamRou-15b} for the non-negative short range 3D case and the short range 2D case to our 3D dipolar problem.
We denote by $h= -\big(\nabla_{x} +iA(x)\big)^2 + V(x)$ the one body operator which is the Friedrichs extension of the operator defined similarly on $C^\infty_0(\mathbb{R}^{3})$. Furthermore, $h$ has a compact resolvent under the assumption that $V$ is confining: $\lim_{|x|\to\infty}V(x) = \infty$ , which we make, see \cite{Iwatsuka-86} for more details on $h$.

The Hartree functional is given by the energy of an condensed state
\begin{align}\label{E_hartree}
\mathcal{E}_H^N(u) &= \frac{\braket{u^{\otimes N}, H_N u^{\otimes N}}}{N} \nonumber \\
  &= \int_{\mathbb{R}^{3}} |(\nabla + i A)u|^2 + \int_{\mathbb{R}^{3}} V |u|^2 + \frac{1}{2} \iint_{\mathbb{R}^{3}\times\mathbb{R}^{3}} (w_N \star |u|^2) |u|^2.
\end{align}
We denote by $$e_{H,N} := \inf_{\|u\|_{L^{2}}=1} \mathcal{E}_{H}^N(u)$$  the Hartree ground state energy.
If a potential $W$ is classically stable then it is Hartree stable \cite{LewNamRou-15}:
\begin{equation}\label{hartree_stab}
\iint_{\mathbb{R}^{3}\times \mathbb{R}^{3}} W(x-y) \rho(x)\rho(y) dx dy \geq 0, \quad \forall \rho \geq 0.
\end{equation}
It is easy to see that Hartree stability is needed in order to have $\lim\inf e_{H,N} > -\infty$. Indeed, if the quantity in (\ref{hartree_stab}) is negative for some $\rho \geq 0$, by taking $u_N(x)=\sqrt{\rho_N(x)} = N^{3\beta/2} \sqrt{\rho(N^\beta x)}$, one can see that $e_{H,N}\to\infty$ as $N\to\infty$. If, besides, $W(0) <\infty$, Hartree stability implies classical stability.

Notice that both in the Hartree functional and in the GP functional, stability is equivalent to the non negativity of the interaction energy, compare Lemma \ref{lem_dip_cont} and condition (\ref{hyp1_theo_GP}) of Theorem \ref{theo_GP}.

The two steps of the proof are as follows. The derivation of the Hartree theory follows arguments of \cite{LewNamRou-13}, using de Finetti theorems combined with a stability argument. The way from Hartree theory to Gross-Pitaevskii theory is essentially based on the observation that for $w\in L^1$, $w_N \to \delta_0\int w$ in the sense of measures.

\subsubsection{Estimating the pair-potential by the kinetic energy}
This section is dedicated to the proof of Proposition \ref{cor_est_w} which is a generalization of \cite[Lemma 3.2]{NamRouSer-15} to a larger class of interactions including the dipolar potential. The latter has singularities at $0$ and infinity and is therefore not a $L^1$ function as in \cite{NamRouSer-15}.
\begin{prop}[Domination of the interaction potential by the kinetic energy]\label{cor_est_w}
Define $w=w_0+\mathds{1}_{|x|\geq R}K$ with $w_0 \in L^1(\mathbb{R}^{3})\cap L^{2}(\mathbb{R}^{3})$, $K$ as in Theorem \ref{theo_1} and $R>0$. Denote $w_N = N^{3\beta} w(N^\beta \cdot)$ for any $0 \leq \beta \leq 1$. We have the following estimates.
\begin{align}
|w_N(x-y)| &\leq C N^\beta h_x \label{cor_est_w_1}\\
\pm w_N(x-y) &\leq C_\varepsilon (h_x)^{3/4+\varepsilon}\otimes (h_y)^{3/4+\varepsilon}, \, \, \forall \epsilon >0, \label{lem_est_w_2} \\
\pm \left( h_x w_N(x-y) + w_N(x-y)h_x \right) &\leq C N^{3\beta/2}  h_x \otimes h_y. \label{cor_est_3}
\end{align}
\end{prop}
We will first state some lemmas to deal with the dipolar potential $K$. The first two have been adapted from \cite{Duo01} and their proof will not be given here. In this paper we use the following Fourier transform $\widehat{f}(k) = \int_{\mathbb{R}^{3}} f(x) e^{-k\cdot x}dx$.
\begin{lem}[Fourier transform of the potential $K$, Corollary 4.5 of \cite{Duo01}] \label{lem_dip_cont}
Let $q>1$ and $\Omega \in L^q(\mathbb{S}^2)$ be a even function satisfying the cancellation property
\begin{equation*}
\int_{\mathbb{S}^2} \Omega(\omega) d\sigma(\omega) = 0.
\end{equation*}
For any $R'>R \geq 0$ we define
\begin{equation}\label{def_K_R}
K_R^{R'}(x) =\mathds{1}_{R'>|x|>R} \frac{\Omega(x/|x|)}{|x|^3}.
\end{equation}
Then for any $R'>R>0$,
\begin{align*}
\widehat{K_R^{R'}} (k) &=  \int_{\mathbb{S}^2} \int_R^{R'} \frac{\cos(r k \cdot \omega) }{r} \Omega(\omega) drd\sigma(\omega)  \\ &= \int_{\mathbb{S}^2} \int_R^{R'} \frac{\cos(r k \cdot \omega) - \cos(r |k|)}{r} \Omega(\omega) drd\sigma(\omega),
\end{align*}
and as $R' \to\infty$ and $R\to0$ it converges for $k\neq 0 $ to
\begin{equation}\label{expr_K_hat}
\widehat{K}(k) = \int_{\mathbb{S}^2} \log\left(\frac{|k|}{|k\cdot \omega|}\right) \Omega(\omega) d\sigma(\omega).
\end{equation}
Moreover 
\begin{equation*}
\|\widehat{K}_R^{R'}\|_{L^\infty(\mathbb{R}^{3})} \leq C_q \|\Omega\|_{L^q(\mathbb{S}^2)}
\end{equation*}
for some constant $C_q>0$ independent of $R'>R \geq 0$ and therefore
\begin{equation*}
\|\widehat{K}\|_{L^\infty(\mathbb{R}^{3})} \leq C_q \|\Omega\|_{L^q(\mathbb{S}^2)}
\end{equation*}
Note that the function $\widehat{K}$ defined in (\ref{expr_K_hat}) is homogeneous of degree zero and continuous except at the origin.
\end{lem}
\begin{lem}[Definition and continuity in $L^p$, Theorem 4.12 of \cite{Duo01}]
Under the assumptions of the previous lemma, for any $1<p<\infty$, we define for $f \in L^p(\mathbb{R}^{3})$
$$
K\star f (x) := \lim_{\substack{R' \to \infty \\ R\to 0 }} K_R^{R'} \star f(x)
$$
which exists for almost every $x\in\mathbb{R}^{3}$. Then, there exists some constant $C_p>0$, independent of $0\leq R < R' \leq \infty$ such that
\begin{equation*}
\|K_R^{R'} \star f\|_{L^p(\mathbb{R}^{3})} \leq C_p C_q \|\Omega\|_{L^q(\mathbb{S}^2)} \|f\|_{L^p(\mathbb{R}^{3})}, \quad \forall f\in L^p(\mathbb{R}^3).
\end{equation*}
\end{lem}
Now we prove an inequality allowing to control long range interactions having bounded Fourier transform, by the kinetic energy.

\begin{lem}[Domination of the long range potential by the kinetic energy]\label{lem_dip_est}
Let $W\in L^2(\mathbb{R}^{3})$ be such that $\widehat{W}\in L^{\infty}(\mathbb{R}^{3})$.
Then, for all $\varepsilon > 0$, there exists a constant $C_{\varepsilon}>0$, independent of $W$, such that the following operator inequality holds 
\begin{equation}\label{ineq_dip}
\pm W(x-y) \leq C_\varepsilon \|\widehat{W}\|_{L^\infty(\mathbb{R}^{3})} (1-\Delta_x)^{3/4 +\varepsilon} \otimes (1-\Delta_y)^{3/4 +\varepsilon}.
\end{equation}
\end{lem}

\begin{proof}
To prove (\ref{ineq_dip}) it is sufficient to show that for any $R>0$ we have
\begin{equation}\label{eq_proof_dip_est}
(1-\Delta_x)^{-\frac{3}{8}-\frac{\varepsilon}{2}}(1-\Delta_y)^{-\frac{3}{8}-\frac{\varepsilon}{2}} W(x-y) (1-\Delta_x)^{-\frac{3}{8}-\frac{\varepsilon}{2}} (1-\Delta_y)^{-\tfrac{3}{8}-\frac{\varepsilon}{2}} \leq C_\varepsilon \|\widehat{W}\|_{L^\infty(\mathbb{R}^{3})}
\end{equation}
with $C_\varepsilon$ independent of $W$. 
We define
\begin{equation*}
G = \mathcal{F}^{-1}\left(\frac{1}{(1+p^2)^{3/8 + \varepsilon/2}}\right),
\end{equation*}
where we denote by $\mathcal{F}^{-1}$ the inverse Fourier transform.
Then, inequality (\ref{eq_proof_dip_est}) is equivalent to
\begin{multline}\label{eq_proof_dip_est2}
\bigg| \int f(x,y) G(x-x') G(y-y') \frac{W(x'-y')}{\|W\|_{L^\infty(\mathbb{R}^{3})}} \times \\ \times G(x'-x'') G(y'-y'') g(x'',y'') dx\,dy\,dx'dy'dx''dy'' \bigg|  \leq C \|f\|_{L^2(\mathbb{R}^{3})} \|g\|_{L^2(\mathbb{R}^{3})},
\end{multline}
for all $f,g\in L^2(\mathbb{R}^3\times \mathbb{R}^{3})$ and for some constant $C>0$. We will prove inequality (\ref{eq_proof_dip_est2}) with $Y \in S$, the Schwartz space, instead of $G$ and we will conclude by a density argument.

Let $f,g\in L^2(\mathbb{R}^3\times \mathbb{R}^{3})$ and let $Y \in \mathcal{S}$ then we have
\begin{align}\label{eq_proof_dip_est3}
&  \int f(x,y) Y(x-x') Y(y-y') W(x'-y') Y(x'-x'')\times \\ 
&\qquad \qquad\qquad \qquad\qquad\qquad \times Y(y'-y'') g(x'',y'') {dx\,dy \,dx'dy' dx'' dy''}\nonumber
\\
& \qquad= \int \bigg( W(w)  \int f(x,y) Y(x-(w+y')) Y(y-y') \times\nonumber \\ 
&\qquad \qquad\qquad \qquad\qquad\qquad  \times Y((w+y')-x'') Y(y'-y'') g(x'',y'')
	{dx\,dy \,dy' dx'' dy''\bigg) dw} \nonumber
 \\
&  \qquad= \int W(w) h(w) dw \nonumber\\
& \qquad= \int \overline{\widehat{W}(p)} \widehat{h}(p)dp,\nonumber
\end{align}
where we have made the change of variable $x' = w + y'$ and we have defined $$h(w) =  \int f(x,y) Y(x-(w+y')) Y(y-y') Y((w+y')-x'') Y(y'-y'') g(x'',y'') \: dxdy dy' dx'' dy''.$$
Since $\|\widehat{W}\|_{L^\infty}<\infty$, it suffices to prove that $\widehat{h}\in L^1(\mathbb{R}^3)$ to get the result. We then compute
\begin{align*}
\widehat{h}(p)
= &\int f(x,y)Y(y-y') \left(\int \widehat{Y}(q)e^{iq\cdot(y'-x)} \widehat{Y}(p-q)e^{i(p-q)\cdot(y'-x'')} dq\right) \times \\
 &\qquad \qquad\qquad \qquad\qquad\qquad \times  Y(y'-y'') g(x'',y'') dx\,dy\, dy' dx''dy'' \\
= & \int  f(x,y)  \widehat{Y}(q)e^{-iq\cdot x} \widehat{Y}(p-q) e^{-i (p-q)\cdot x''} \left(\int \widehat{Y}(k)e^{ik\cdot y} \widehat{Y}(p-k)e^{i(p-k)\cdot y''} dk\right)\times \\
& \qquad \qquad\qquad \qquad\qquad\qquad \times g(x'',y'') dq\, dx\,dy\, dx''dy''\\
= &\int  \widehat{f}(q,-k) \widehat{g}(q-p,k-p) \widehat{Y}(q) \widehat{Y}(p-q) \widehat{Y}(k)\widehat{Y}(p-k) \widehat{g}(q-p,k-p) dk\,dq
\end{align*}
We then have
\begin{align*}
\int |\widehat{h}(p)| dp &\leq
\int \left| \widehat{f}(q,-k) \widehat{g}(q-p,k-p) \widehat{Y}(q) \widehat{Y}(p-q) \widehat{Y}(k) \widehat{Y}(p-k) \right| dpdqdk \\
&\leq \frac{1}{2} \int \left| \widehat{f}(q,-k) \widehat{g}(q-p,k-p)\right| \left( |\widehat{Y}|^2(q)|\widehat{Y}|^2(p-q) + |\widehat{Y}|^2(k) |\widehat{Y}|^2(p-k)\right)dpdqdk \\
&\leq \frac{1}{2}\Bigg[ \left(\int  |\widehat{f}|^2(q,-k)|\widehat{Y}|^4(p-q)dpdqdk\right)^{1/2} \left(\int |\widehat{g}|^2(q-p,k-p) |\widehat{Y}|^4(q)dpdqdk\right)^{1/2} \\ 
&\quad + \left(\int |\widehat{f}|^2(q,-k) |\widehat{Y}|^4(p-k)dpdqdk\right)^{1/2} \left( \int |\widehat{g}|^2(q-p,k-p) |\widehat{Y}|^4(k)dpdqdk\right)^{1/2}\Bigg]\\
&\leq  \|f\|_{L^2(\mathbb{R}^{6})} \|g\|_{L^2(\mathbb{R}^{6})} \|\widehat{Y}\|_{L^4(\mathbb{R}^{3})}^4 \\
&\leq C_{4/3}  \|f\| \|g\| \|Y\|_{L^{4/3}(\mathbb{R}^{3})}^4 .
\end{align*}

Where $C_{4/3}$ is the Lipschitz modulus of the Fourier transform from $L^{4/3}(\mathbb{R}^3)$ to $L^4(\mathbb{R}^3)$. We have just proved that the expression (\ref{eq_proof_dip_est3}) defines a continuous $4$-linear form on $L^{4/3}(\mathbb{R}^3)$ by density of smooth functions. Recalling that $G\in L^{4/3}(\mathbb{R}^3)$ proves (\ref{eq_proof_dip_est2}) and ends the proof.
\end{proof}


Thanks to Lemma \ref{lem_dip_est} and Lemma \ref{lem_magn_lapl} we are now able to prove Proposition \ref{cor_est_w}.

\begin{proof}[Proof of Proposition \ref{cor_est_w}]
The first inequality (\ref{cor_est_w_1}) is exactly the same as in \cite[Lemma 3.2]{NamRouSer-15} and the last inequality (\ref{cor_est_3}) is easily adapted from \cite{NamRouSer-15}. The major improvement concerns (\ref{lem_est_w_2}) where in \cite{NamRouSer-15} the potential was assumed to be in $L^1$ which the dipolar potential is not.\\

\subsubsection*{Proof of (\ref{cor_est_w_1}).} 
Using the Sobolev inequality one has
\begin{equation*}
|W(x-y)| \leq C \|W\|_{L^{3/2}} (-\Delta_x),
\end{equation*}
for any $W\in L^{3/2}(\mathbb{R}^{3})$. Then using $\|w_N\|_{L^{3/2}} = N^\beta \|w\|_{L^{3/2}}$ and Lemma \ref{lem_magn_lapl} we obtain (\ref{cor_est_w_1}).

\subsubsection*{Proof of (\ref{lem_est_w_2}).} 
Let us write $w(x) = w_0(x) + \mathds{1}_{|x|>R}K(x)$ with $w_0\in L^{1}$. Then, $\|\widehat{w}\|_{L^\infty} \leq \|\widehat{w_0}\|_{L^{\infty}} + \|\widehat{K_R}\|_{L^\infty} \leq \|w_0\|_{L^1} + C_q \|\Omega\|_{L^q(\mathbb{S}^2)}$ thanks to Lemma \ref{lem_dip_cont}. We can now apply Lemma \ref{lem_dip_est} which proves (\ref{lem_est_w_2}).\\

\subsubsection*{Proof of (\ref{cor_est_3}).} 
The proof of \cite[Lemma 3.2]{NamRouSer-15} gives
\begin{equation*}
\pm \left( h_x W(x-y) + W(x-y) h_x \right) \leq C(\|W\|_{L^{3/2}} + \|W\|_{L^2}) h_x\otimes h_y.
\end{equation*}
and since $\|w_N\|_{L^{3/2}} = N^\beta\|w\|_{L^{3/2}}$ and $\|w_N\|_{L^{2}} = N^{3\beta/2} \|w\|_{L^{2}}$, inequality (\ref{cor_est_3}) follows.
\end{proof}
\subsection{From the many-body problem to Hartree theory}

Our proof is inspired of the one for the $2D$ case of \cite{LewNamRou-15b} with techniques first introduced in \cite{LewNamRou-15}. The method is the following: we derive a lower bound on the many-body energy per particle $e_N$ making two approximations. First we project the ground state on a finite (but varying) number of low energy levels of the one-body operator, and then we use the quantitative de Finetti theorem in finite dimension to approximate the projection of the ground state by a Hartree state. The error terms involve $L$ (the dimension of the low energy subspace the ground state has been projected on), $N$ the number of particles and $e_{N,\varepsilon}$ the energy per particle itself (with kinetic energy decreased by a factor $(1-\varepsilon)$) which is not yet known to be bounded. But, by a bootstrap argument, first introduced in \cite{LewNamRou-15b}, we get a condition on $\beta$ ensuring the boundedness of $e_N$. As in \cite{NamRouSer-15,LewNamRou-15b} the main improvement regarding \cite{LewNamRou-15} is the use of moment estimates on ground states rather than pure operator estimates.

\subsubsection{Moment estimates}

Now, we prove some moment estimates which will be useful to control the errors. Namely, we adapt \cite[Lemma 5]{LewNamRou-15b} to the 3D case, the main difference being the assumption $0<\beta<1$ which has to be strengthened due to the negative part of the potential $w$.
We define the Hamiltonian
$$
H_{N,\varepsilon} := H_N - \varepsilon \sum_{j=1}^N h_j
$$
where $h$ is replaced by $(1-\varepsilon)h$, and we denote by $N \times e_{N,\varepsilon}$ its ground state energy.

\begin{lem}[Moments estimates]\label{lem_mom_est}
Let $0<\beta<2/3$ and $\Psi_N \in \mathfrak{h}^N$ be a ground state of $H_N$. Then for all $\varepsilon \in (0,1)$ we have
\begin{equation}\label{moment_estimate_1}
\tr\left( h \gamma^{(1)}_{\Psi_N}\right) \leq C \frac{1+|e_{N,\varepsilon}|}{\varepsilon}
\end{equation}
and
\begin{equation}\label{moment_estimate_2}
\tr\left( h\otimes h \gamma^{(2)}_{\Psi_N}\right) \leq C \frac{(1+|e_{N,\varepsilon}|)^2}{\varepsilon^2},
\end{equation}
for $N$ sufficiently large, and where the constant $C$ is independent of $N$.
\end{lem}

\begin{proof}
We have
\begin{equation}\label{H_N,e}
H_{N,\varepsilon} = H_N - \varepsilon \sum_{j=1}^N h_j \geq N e_{N,\varepsilon}
\end{equation}
since $H_N \Psi_N = N e_N \Psi_N$, then by taking the trace of (\ref{H_N,e}) against $\gamma=\ket{\Psi_N}\bra{\Psi_N}$ we get
\begin{align*}
e_N - \varepsilon \tr(h \gamma) \geq e_{N,\varepsilon}
\end{align*}
which shows the first moment estimate since $e_N$ is upper bounded by a constant. For the second inequality, let us write
\begin{align}\label{est_w_1}
\frac{1}{N^2} \left< \Psi_N, \left( \left(\sum_{j=1}^N h_j \right) H_N + H_N \left( \sum_{j=1}^N h_j\right) \right) \Psi_N \right> \nonumber\\
= 2 \frac{e_N}{N} \left< \Psi_N, \sum_{j=1}^N h_j \Psi_N \right> \leq C \frac{(1+|e_{N,\varepsilon}|)^2}{\varepsilon} 
\end{align}
where we used that $|e_N|\lesssim 1+ |e_{N,\varepsilon}|$ and the first moment estimate. Let us find a lower bound on (\ref{est_w_1}). One has
\begin{multline}\label{est_w_2}
\frac{1}{N^2} \left[ \left(\sum_{j=1}^N h_j \right) H_N + H_N \left( \sum_{j=1}^N h_j\right)\right] \\
=  \frac{2}{N^2} \left(\sum_{j=1}^N h_j \right)^2 + \frac{1}{N^2(N-1)} \sum_{i=1}^N \sum_{j<k}\left( h_i w_N(x_j-x_k) + w_N(x_j-x_k)h_i \right).
\end{multline}
We split the second term in (\ref{est_w_2}) in two terms, depending whether $i\neq j <  k \neq i$ or $i \in \{j,k\}$. For any $1\leq i_0 \leq N$, we have
\begin{align}\label{est_w}
\frac{1}{N-1} \sum_{\substack{1\leq j<k\leq N \\ j\neq i_0, k\neq i_0}} w_N(x_j-x_k) &= H_{N,\varepsilon} - (1-\varepsilon) \sum_{j=1}^N h_j - \frac{1}{N-1}\sum_{j\neq i_0}w_N(x_{i_0}-x_j)\nonumber \\
&\geq N e_{N,\varepsilon} - \left(1-\varepsilon + C \frac{N^\beta}{N-1}\right) \left(\sum_{j=1}^N h_j \right),
\end{align}
where (\ref{cor_est_w_1}) of Lemma \ref{cor_est_w} has been used. Then multiplying (\ref{est_w}) by $h_i$, we obtain
\begin{multline}\label{est_w_fin1}
\frac{1}{N^2(N-1)}\sum_{i=1}^N \sum_{i\neq j<k\neq i}  h_i w_N(x_j-x_k) + w_N(x_j-x_k)h_i \\
\geq 2 \frac{e_{N,\varepsilon}}{N} \left(\sum_{j=1}^N h_j \right) - \frac{2}{N^2} \left(1-\varepsilon + C \frac{N^\beta}{N}\right) \left(\sum_{j=1}^N h_j \right)^2.
\end{multline}
On the other hand, for $j\neq k$, we have by (\ref{cor_est_3})
\begin{equation*}
h_j w_N(x_j-x_k) + w_N(x_j-x_k) h_j \geq - C N^{3\beta/2} h_j h_k
\end{equation*}
and, after summing for $1 \leq j<k \leq N$, we obtain
\begin{equation}\label{est_w_fin2}
\sum_{j<k} h_j w_N(x_j-x_k) + w_N(x_j-x_k) h_j \geq - C N^{3\beta/2} \left(\sum_{j=1}^N h_j \right)^2.
\end{equation}
Combining (\ref{est_w_fin1}) and (\ref{est_w_fin2}) we arrive at the estimate
\begin{align}\label{est_w_fin}
\frac{1}{N^2} \left(\sum_{j=1}^N h_j \right) H_N + H_N \left( \sum_{j=1}^N h_j\right) \geq &\frac{2}{N^2} \left(\varepsilon - C \frac{N^\beta}{N}- C \frac{N^{3\beta/2}}{N}\right)  \left(\sum_{j=1}^N h_j \right)^2  \nonumber \\ 
& - C\frac{1+|e_{N,\varepsilon}|}{N} \left(\sum_{j=1}^N h_j \right).
\end{align}
Taking the trace of (\ref{est_w_fin}) against $\ket{\Psi_N}\bra{\Psi_N}$ and using the first moment estimate, we conclude that
\begin{equation*}
\left( \varepsilon - C \frac{N^\beta}{N}- C \frac{N^{3\beta/2}}{N}  \right) \tr\left( h\otimes h \gamma^{(2)}_{\Psi_N} \right) \leq C \frac{(1+|e_{N,\varepsilon}|)^2}{\varepsilon}.
\end{equation*}
Now, if $\beta < 2/3$, for $N$ large enough one has $\varepsilon - C N^{\beta-1}- C N^{3\beta/2-1} \geq \varepsilon/2 >0$ and the second moment estimate is proved.
\end{proof}

\subsubsection{Lower bound via de Finetti}
We now state some quantitative version of the de Finetti theorem \cite[Lemma 3]{LewNamRou-15b} originally proven in \cite{ChrKonMitRen-07} (see also \cite{Chi-10,Har-13,LewNamRou-14} for variants of the proof and \cite{FanVan-06} for new results). For a summary of the use of de Finetti theorems in the mathematics of ultra-cold atomic gases, see \cite{Rou-15}.
Heuristically, the result states that the density matrices of symmetric wave functions are well approximated by convex combinations of product states when $N$ is large.
\begin{theo}[Quantitative quantum de Finetti in finite dimension]\label{theo_definetti}
Let $\Psi \in \mathfrak{H}^N = \bigotimes_s^N L^{2}(\mathbb{R}^{3})$ and $P$ be an orthogonal projection of finite rank. Then, there exists a positive Borel measure $d\mu_\Psi$ on the unit sphere $S P\mathfrak{H}$ such that
\begin{equation*}
\tr_{\mathfrak{H}} \left| P^{\otimes 2}\gamma_{\Psi}^{(2)}P^{\otimes 2} - \int_{SP\mathfrak{H}} \ket{u^{\otimes 2}}\bra{u^{\otimes 2}} d\mu_{\Psi}(u)\right| \leq \frac{8 \dim P\mathfrak{H}}{N}
\end{equation*}
and
\begin{equation*}
 \int_{SP\mathfrak{H}} d\mu_\Psi(u) \geq \left(\tr P\gamma_{\Psi}^{(1)} \right)^2.
\end{equation*}
\end{theo}

We then denote by $K_2^N = h \otimes 1 + 1 \otimes h +\frac{1}{2} w_N(x-y)$ and $P := P(L)= \mathds{1}_{(-\infty, L]}(h)$, for any $L>0$, the projection onto the subspace of energy levels lower than $L$ of the one-particle operator. The ground state energy can be written
\begin{align*}
\frac{\braket{\Psi_N,H_N \Psi_N}}{N} = \tr\left( K_2 \gamma^{(2)}_{\Psi_N} \right) = \int_{\mathbb{S}L^2} \braket{u^{\otimes 2} , K_2 u^{\otimes 2}} d\mu_{\Psi_N}(u) 
+ \tr K_2 \left(\gamma_{\Psi_N}^{(2)} - P^{\otimes 2}\gamma_{\Psi_N}^{(2)} P^{\otimes 2}  \right) \nonumber \\ 
+ \tr K_2 \left( P^{\otimes 2} \gamma_{\Psi_N}^{(2)} P^{\otimes 2} - \int_{\mathbb{S}L^2} \ket{u^{\otimes 2}}\bra{u^{\otimes 2}} d\mu_{\Psi_N}(u)) \right).
\end{align*}
Where we have denoted the unit sphere of $L^2(\mathbb{R}^{3})$ by $SL^2$. The first term is the sought after approximation leading to the Hartree energy :
\begin{equation*}
\mathcal{E}_H^N(u) = \braket{u^{\otimes 2},K_2^N u^{\otimes 2}}.
\end{equation*}
The two others are error terms which have to be controlled. We summarize our estimates in the following lemma.

\begin{lem}[Lower bound via de Finetti]\label{cor_est_de_fin}
For any $0<\beta \leq 1$ and any $L,\varepsilon, \delta >0$, there exists a constant $C>0$ and $N_0> 0$ such that, we have for all $N\geq N_0$
\begin{equation}\label{est_de_fin_1}
\int_{\mathbb{S}L^2} \braket{u^{\otimes 2} , K_2 u^{\otimes 2}} d\mu_{\Psi_N}(u)
 \geq e_{H,N} - C\frac{(1+|e_{N,\varepsilon}|)}{\varepsilon L},
\end{equation}
\begin{equation}\label{est_de_fin_2}
\tr \left( K_2 \left( P^{\otimes 2}\gamma_{\Psi_N}^{(2)}P^{\otimes 2} - \int_{\mathbb{S}L^2} \ket{u^{\otimes 2}}\bra{u^{\otimes 2}} d\mu_{\Psi_N}(u)\right)  \right)
\geq -  C_{\delta} \frac{L^{3(1+\tfrac{1}{s})+2\delta}}{N}
\end{equation}
and
\begin{equation}\label{est_de_fin_3}
\tr K_2 \left(\gamma_{\Psi_N}^{(2)} - P^{\otimes 2}\gamma_{\Psi_N}^{(2)} P^{\otimes 2}  \right) 
\geq - C_\delta \frac{(1+|e_{N,\varepsilon}|)^{\tfrac{13}{8} + \tfrac{3\delta}{2}}}{\varepsilon L^{\tfrac{1}{8}-\tfrac{\delta}{2}}}.
\end{equation}
Combining the inequalities (\ref{est_de_fin_1}), (\ref{est_de_fin_2}) and (\ref{est_de_fin_3}) we find
\begin{equation*}
e_N \geq e_{H,N} - C\frac{(1+|e_{N,\varepsilon}|)}{\varepsilon L} - C_{\delta} \left(\frac{L^{3(1+\tfrac{1}{s}) + 2\delta}}{N} + \frac{(1+|e_{N,\varepsilon}|)^{\tfrac{13}{8} + \tfrac{3\delta}{2}}}{\varepsilon^{\tfrac{13}{8} + \tfrac{3\delta}{2}} L^{\tfrac{1}{8}-\tfrac{\delta}{2}}}\right).
\end{equation*}
\end{lem}
\begin{proof}Let us take $0<\beta \leq 1$, $L,\varepsilon, \delta >0$ and $N_0>0$ such that Lemma \ref{lem_mom_est} holds.
\subsubsection*{Main term (\ref{est_de_fin_1}).}
Since $\braket{u^{\otimes 2}, K_2 u^{\otimes 2}} = \mathcal{E}_H(u)$, we have
\begin{equation*}
 \int_{\mathbb{S}L^2} \braket{u^{\otimes 2} , K_2 u^{\otimes 2}} d\mu_{\Psi_N}(u) \geq e_{H,N}  \int_{\mathbb{S}L^2} d\mu_{\Psi_N}(u).
\end{equation*}
But $ \int_{\mathbb{S}L^2} d\mu_{\Psi_N}(u) \leq \left(\tr P\gamma_{\Psi_N}^{(1)}\right)^2$ and $Q:=1-P \leq L^{-1} h$ and therefore 
\begin{align*}
 \int_{\mathbb{S}L^2} d\mu_{\Psi_N}(u) &\geq (1- \tr Q \gamma_{\Psi_N}^{(1)})^2 \geq 1 - \frac{2}{L}\tr h \gamma_{\Psi_N}^{(1)}.
\end{align*}
So we obtain
\begin{equation}\label{main_term_fin}
 \int_{\mathbb{S}L^2} \braket{u^{\otimes 2} , K_2 u^{\otimes 2}} d\mu_{\Psi_N}(u) \geq e_{H,N} - \frac{C}{L} \tr h \gamma_{\Psi_N}^{(1)} 
 \geq e_{H,N} - C\frac{(1+|e_{N,\varepsilon}|)}{\varepsilon L}
\end{equation}
where we have used (\ref{moment_estimate_1}).\\

\subsubsection*{First error term (\ref{est_de_fin_3}).}
Using the quantitative quantum de Finetti Theorem \ref{theo_definetti} and the bound $P h \leq L P$ we obtain
\begin{equation}\label{first_error_1}
\left| \tr \left( h_1 + h_2 \right) \left( P^{\otimes 2} \gamma_{\Psi_N}^{(2)} P^{\otimes 2} -  \int_{\mathbb{S}L^2} \ket{u^{\otimes 2}}\bra{u^{\otimes 2}} d\mu_{\Psi_N}(u) \right) \right| \leq C \frac{L d_L}{N},
\end{equation}
where $d_L=\dim \ima P$. Besides, from (\ref{lem_est_w_2}) in Corollary \ref{cor_est_w} we deduce that for $\delta > 0$,
\begin{equation*}
\pm P^{\otimes 2} w_N(x_1-x_2) P^{\otimes 2} \leq C_{\delta} (Ph_1 \otimes Ph_2)^{3/4 + \delta} \leq C_{\delta} L^{3/2+2\delta } P^{\otimes 2}
\end{equation*}
which, combined with (\ref{first_error_1}), gives
\begin{align}\label{first_erro_fin}
\tr \left( K_2 \left( P^{\otimes 2}\gamma_{\Psi_N}^{(2)}P^{\otimes 2} - \int_{\mathbb{S}L^2} \ket{u^{\otimes 2}}\bra{u^{\otimes 2}} d\mu_{\Psi_N}(u)\right)  \right)
&\geq -  C_{\delta} \frac{L^{3/2+2\delta} \dim P \mathfrak{H}}{N}. 
\end{align}
We now use \cite[Theorem 2.1]{ComSchSei-78} according to which there is some constant $C>0$ such that $\dim P \mathfrak{H} \leq C L^{3(1/2+1/s)}$ and which shows (\ref{est_de_fin_2}).\\

\subsubsection*{Second error term (\ref{est_de_fin_2}).}
The operator inequality $h \geq P h$ gives
\begin{equation*}
\tr \left( \left( h_1 + h_2 \right) \left(\gamma_{\Psi_N}^{(2)}-P^{\otimes 2}\gamma_{\Psi_N}^{(2)}P^{\otimes 2}\right)\right) 
= \tr \left[ \left( 	\left(h_1+ h_2\right) - P^{\otimes 2}\left(h_1 + h_2\right)P^{\otimes 2} \right) \gamma_{\Psi_N}^{(2)} \right]
\geq 0.
\end{equation*}
Then, we follow and adapt \cite{LewNamRou-15b}. Using the Cauchy-Schwarz inequality for operators 
\begin{equation}\label{CS_op}
\pm \left( AB + B^* A^*\right) \leq \eta^{-1} A A^* + \eta B^* B, \quad \forall \eta>0,
\end{equation}
we obtain
\begin{align}\label{sec_err_1}
  &\pm 2 \left( \gamma_{\Psi_N}^{(2)}-P^{\otimes 2}\gamma_{\Psi_N}^{(2)}P^{\otimes 2}\right) \nonumber \\
= &\pm \left( \left(1-P^{\otimes 2} \right) \gamma_{\Psi_N}^{(2)} + \gamma_{\Psi_N}^{(2)}\left(1-P^{\otimes 2}\right) 
+ P^{\otimes 2}\gamma_{\Psi_N}^{(2)}\left(1-P^{\otimes 2}\right) 
+ (1-P^{\otimes 2})\gamma_{\Psi_N}^{(2)}P^{\otimes 2} \right) \nonumber \\
\leq & \, 2 \eta^{-1}\left(1-P^{\otimes 2}\right)\gamma_{\Psi_N}^{(2)}\left(1-P^{\otimes 2}\right) 
+ \eta \left(\gamma_{\Psi_N}^{(2)}+P^{\otimes 2}\gamma_{\Psi_N}^{(2)}P^{\otimes 2}\right), \quad \forall \eta>0.
\end{align}
On the other hand, using (\ref{lem_est_w_2}) and the fact that
\begin{equation}\label{inf_trick}
t^r = \inf_{\substack{\\\rho>0}}\left(r \rho^{-1} t + (1-r) \rho^{\frac{r}{1-r}} \right), \quad \forall t\geq 0, \quad \forall r\in (0,1),
\end{equation}
we get
\begin{equation}\label{sec_err_2}
\pm w_N(x-y) \leq C_{\delta} (h_xh_y)^{3/4+\delta} \leq C_\delta \left( \rho^{-1} h_xh_y + \rho^{\frac{3+4\delta}{1-4\delta}} \right).
\end{equation}
Now, combining (\ref{sec_err_1}) and (\ref{sec_err_2}) we have
\begin{align*}
2 \tr \left(w_N \left( \gamma_{\Psi_N}^{(2)}-P^{\otimes 2}\gamma_{\Psi_N}^{(2)}P^{\otimes 2}\right)\right)
\geq &-2 C_\delta \eta^{-1} \tr \left( (h_xh_y)^{3/4+\delta} \right) \left(\left(1-P^{\otimes 2}\right)\gamma_{\Psi_N}^{(2)}\left(1-P^{\otimes 2}\right)\right)  \\
	&-C_\delta \eta \tr \left( \rho^{-1} h_x h_y + \rho^{\frac{3+4\delta}{1-4\delta}} \right) \left( \gamma_{\Psi_N}^{(2)}+P^{\otimes 2}\gamma_{\Psi_N}^{(2)}P^{\otimes 2}\right).
\end{align*}
After optimizing over $\eta>0$ we find
\begin{align}\label{sec_err_3}
\tr \left(w_N \left( \gamma_{\Psi_N}^{(2)}-P^{\otimes 2}\gamma_{\Psi_N}^{(2)}P^{\otimes 2}\right)\right)
\geq - C_\delta &\left[\tr \left( (h_xh_y)^{3/4+\delta} \right) \left(\left(1-P^{\otimes 2}\right)\gamma_{\Psi_N}^{(2)}\left(1-P^{\otimes 2}\right)\right)\right]^{1/2} \\
\times & \left[\tr \left( \rho^{-1} h_x h_y + \rho^{\frac{3+4\delta}{1-4\delta}} \right) \left( \gamma_{\Psi_N}^{(2)}+P^{\otimes 2}\gamma_{\Psi_N}^{(2)}P^{\otimes 2}\right)\right]^{1/2}. \nonumber
\end{align}
Let us deal with the second factor in (\ref{sec_err_3}), we have
\begin{align*}
 \left|\tr \left( \rho^{-1} h_x h_y + \rho^{\frac{3+4\delta}{1-4\delta}} \right) \left( \gamma_{\Psi_N}^{(2)}+P^{\otimes 2}\gamma_{\Psi_N}^{(2)}P^{\otimes 2}\right)\right|
\leq C \left( \rho^{-1} \tr \left(h_1 h_2 \gamma_{\Psi_N}^{(2)}\right) + \rho^{\frac{3+4\delta}{1-4\delta}} \right)
\end{align*}
and by optimizing over $\rho>0$, we obtain
\begin{equation}\label{sec_err_8}
 \left|\tr \left( \rho_0^{-1} h_x h_y + \rho_0^{\frac{3+4\delta}{1-4\delta}} \right) \left( \gamma_{\Psi_N}^{(2)}+P^{\otimes 2}\gamma_{\Psi_N}^{(2)}P^{\otimes 2}\right)\right|
\leq C \tr \left(h_1 h_2 \gamma_{\Psi_N}^{(2)}\right)^{3/4 + \delta}.
\end{equation}
Now, let us find an upper bound of the first factor in (\ref{sec_err_3}). We define $Q:= 1-P$ and we notice that 
\begin{align}\label{sec_err_4}
1 - P^{\otimes 2} &= 1 - (1-Q)\otimes(1-Q) 
=  Q \otimes 1 + 1 \otimes Q - Q \otimes Q \nonumber \\
&\leq Q \otimes 1 + 1 \otimes Q.
\end{align}
This allows us to write
\begin{align}
\left(1 - P^{\otimes 2}\right) \left(h\otimes h\right)^{3/4+\delta} \left(1 - P^{\otimes 2}\right)
&\leq C \left( Qh^{3/4+\delta}\otimes h^{3/4+\delta} + h^{3/4+\delta} \otimes Q h^{3/4+\delta} \right) \label{sec_err_5}\\
&\leq C \frac{1}{L^{1/4-\delta}} \left(h\otimes h^{3/4+\delta} + h^{3/4+\delta}\otimes h \right)\label{sec_err_6} \\
&\leq C \frac{1}{L^{1/4-\delta}} \left( \eta^{-1} h \otimes h + \eta^{\frac{3+4\delta}{1-4\delta}}\left(h\otimes 1 + 1 \otimes h \right)\right) \label{sec_err_7},
\end{align}
for all $\eta >0$. In (\ref{sec_err_5}), we used (\ref{sec_err_4}), whereas in (\ref{sec_err_6}) we used that $Q \leq L^{-1} h$. Finally, in (\ref{sec_err_7}) we used (\ref{inf_trick}). Now, taking the trace of (\ref{sec_err_7}) against $\gamma_{\Psi_N}^{(2)}$ and optimizing over $\eta>0$ we get
\begin{equation}\label{sec_err_9}
\:\mathllap\tr \left( (h_xh_y)^{3/4+\delta}  \left(1-P^{\otimes 2}\right)\gamma_{\Psi_N}^{(2)}\left(1-P^{\otimes 2}\right) \right) \leq \frac{1}{L^{\tfrac{1}{4}-\delta}} \left(\tr h\otimes h \gamma_{\Psi_N}^{(2)}\right)^{3/4+\delta} \left(\tr h\gamma_{\Psi_N}^{(1)}\right)^{1/4-\delta}.
\end{equation}
Finally, from (\ref{sec_err_8}) and (\ref{sec_err_9}) we get
\begin{equation}\label{sec_err_fin}
\tr \left(w_N \left( \gamma_{\Psi_N}^{(2)}-P^{\otimes 2}\gamma_{\Psi_N}^{(2)}P^{\otimes 2}\right)\right)
\geq - \frac{C_\delta}{L^{1/8-\delta/2}}  \left(\tr h\otimes h \gamma_{\Psi_N}^{(2)}\right)^{3/4+\delta/2} \left(\tr h\gamma_{\Psi_N}^{(1)}\right)^{1/8-\delta/2}.
\end{equation}
Now, using the moment estimates (\ref{moment_estimate_1}) and (\ref{moment_estimate_2}) of Lemma \ref{lem_mom_est} the result follows.
\end{proof}
\subsubsection{Final energy estimates and stability of the second kind}
We follow the bootstrap argument of \cite{LewNamRou-15b}. From Lemma \ref{cor_est_de_fin} we have
\begin{equation}\label{fin_ene_est1}
e_{H,N} \geq e_N \geq e_{H,N} - C\frac{(1+|e_{N,\varepsilon}|)}{\varepsilon L} - C_{\delta} \left(\frac{L^{3(1+s^{-1}) + 2\delta}}{N} + \frac{(1+|e_{N,\varepsilon}|)^{13/8 + 3\delta/2}}{\varepsilon^{13/8 + 3\delta/2} L^{1/8-\delta/2}}\right).
\end{equation}
We want to prove that $e_{N,\varepsilon}$ is bounded, otherwise, since we take $L \geq 1$ we have
\begin{equation}\label{supp_term}
\frac{(1+|e_{N,\varepsilon}|)}{\varepsilon L} \leq 
 \frac{(1+|e_{N,\varepsilon}|)^{13/8 + 3\delta/2}}{\varepsilon^{13/8 + 3\delta/2} L^{1/8-\delta/2}}.
\end{equation}
We can therefore omit the left side of (\ref{supp_term}) in inequality (\ref{fin_ene_est1}).
Now, let us rewrite (\ref{fin_ene_est1}), replacing $w$ by $(1-\varepsilon)^{-1}w$, we get
\begin{equation}\label{est_e_e_epsilon}
e_{H,N}^\varepsilon \geq e_{N,\varepsilon} \geq e_{H,N}^\varepsilon - C_{\delta} \left(\frac{L^{3(1+\tfrac{1}{s}) + 2\delta}}{N} + \frac{(1+|e_{N,\varepsilon'}|)^{\tfrac{13}{8} + \tfrac{3\delta}{2}}}{(\varepsilon'-\varepsilon)^{\tfrac{13}{8} + \tfrac{3\delta}{2}} L^{\tfrac{1}{8}-\tfrac{\delta}{2}}}\right),
\end{equation}
for all $1 > \varepsilon' > \varepsilon >0$. Note as well that $e_{H,N}^\varepsilon \geq 0$. Thanks to the classical stability of $w$ (\ref{classical_stab}), we know that for all $0\leq \varepsilon \leq 1$ we have ${e_{N,\varepsilon} \geq - C N^{3\beta-1}}$. Now, as in \cite{LewNamRou-15b}, this leads us to make the following induction hypothesis, denoted $I_\eta$, for $\eta\geq 0$,
\begin{equation*}
\limsup_{\substack{\\N\to\infty}} \frac{|e_{N,\varepsilon}|}{1+ N^\eta} < \infty , \, \forall \varepsilon \in (0,1).
\end{equation*}
It is clear that $I_\eta$ holds for $\eta = 3\beta - 1$, and we would like to prove $I_0$. By choosing $L=N^{\tau}$, for $\tau>0$, we see that if $I_\eta$ holds, then $I_{\eta'}$ also holds as soon as 
\begin{equation}\label{choose_eta_1}
\eta' > \max\left\{ 3 \tau (1+s^{-1}) - 1 , \frac{13\eta - \tau}{8} \right\}.
\end{equation}
Optimizing over $\tau$, (\ref{choose_eta_1}) becomes
\begin{equation*}
\eta' > \eta - \frac{s - \eta (15+14s)}{24+25s}.
\end{equation*}
In order to choose some $\eta' < \eta$, the condition $\eta < s/(15+14s)$ must be fulfilled. In term of $\beta$ (in order to start the induction) this means
\begin{equation}\label{cond_beta}
\beta < \frac{1}{3} + \frac{s}{45 + 42s}.
\end{equation}
If condition (\ref{cond_beta}) is fulfilled, we can show that $(I_0)$ holds, after applying (\ref{est_e_e_epsilon}) finitely many times. Namely, we have convergence of the energy per particle with the following estimate on the rate of convergence
\begin{equation}\label{bound_on_CV}
e_N \geq e_{H,N} - C_\delta N^{-\tfrac{1}{(25 + 24/s) +12 \delta}}.
\end{equation}

\subsection{From Hartree theory to Gross-Pitaevskii}

We have shown above that the many-body ground state energy $e_N$ is well approximated by Hartree's energy $e_{H,N}$. It remains to show that $e_{H,N}$ is a good approximation of the Gross-Pitaevskii energy $e_{GP}(a,b)$. This is easier than the approximation of the many body energy by the Hartree energy. The case $w\in L^1(\mathbb{R}^{3})$ follows from standard arguments \cite{LewNamRou-15,LewNamRou-15b} and it remains to adapt the proof to deal with $K$.

\begin{lem}[From Hartree to Gross-Pitaevskii]\label{lem_from_H_to_nls}
Let $w_0 \in L^1(\mathbb{R}^{3})$ and $K(x) = \Omega(x)/ |x|^3$ satisfying the assumptions of Theorem \ref{theo_1}. Let $b\in\mathbb{R}$ and $a=\int_{\mathbb{R}^{3}} w_0$, let us define $w = w_0 + b \mathds{1}_{|x|>R}K$ for some $R>0$. Then
\begin{equation}\label{eq_lem_HTNLS}
\begin{aligned}
\lim_{N \to \infty} \sup_{\substack{u\in H^1\\ u \neq 0}} \| u \|_{H^1}^{-4} \bigg| &\iint w_N(x-y)|u(x)|^2 |u(y)|^2 dxdy  \\ 
&- a \int |u(x)|^4 dx - b \iint K(x-y) |u(x)|^2 |u(y)|^2 dxdy \bigg| = 0 
\end{aligned}
\end{equation}
\end{lem}
We have not stated here any rate of convergence as it depends on the properties of the short range potential $w_0$. The dipolar part of the Hartree energy converges with an error of $O(N^{-\beta})$ towards the dipolar part of the Gross-Pitaevskii energy. As an example, under the extra assumption $|x|w_0(x) \in L^1(\mathbb{R}^{3})$, the short range part of the energy converges also with an error of $O(N^{-\beta})$.

Lemma \ref{lem_from_H_to_nls} shows that $e_{H,N} = e_{GP} + o(1)$ as $N\to\infty$ and ends the proof of Theorem \ref{theo_1}.
\begin{proof}
For the $L^1$ part of the potential, the proof of convergence is the same as in \cite[Lemma 7]{LewNamRou-15b} except that the integration domain is $\mathbb{R}^{3}$ and not $\mathbb{R}^{2}$. We quickly recall it here: let $w_0\in L^1(\mathbb{R}^3)$ and $A > 0$. We have
\begin{multline*}
\left| \iint w_{0,N}(x-y)|u(x)|^2 |u(y)|^2 dxdy - a \int |u(y)|^4 dy \right| \\
= \left|\iint w_{0}(x) |u(y)|^2\left( |u(N^{-\beta}x+y)|^2 - |u(y)|^2\right) dxdy\right| \\
\qquad\qquad\quad \; \leq 2 \iint \mathds{1}_{|x|<A} |w_{0}(x)| |u(y)|^2  \int_0^1| u \nabla u(tN^{-\beta}x+y)\cdot x| N^{-\beta} dt dxdy \\
+ 2\left( \int \mathds{1}_{|x|>A}|w_0(x)| dx\right) \|u\|_{L^4}^4  \\
\leq 2 \left(AN^{-\beta} + 2 \left( \int \mathds{1}_{|x|>A}|w_0(x)| dx\right)\right) \|u\|_{H^1}^4.\qquad\qquad\qquad\qquad\qquad\quad \;\;\,\;
\end{multline*}
Taking $A = N^{\beta/2}$ gives the desired result. For the dipolar part, the same proof will not work since the potential is not integrable both at the origin and at infinity. Nevertheless, as we saw in Lemma \ref{lem_dip_cont}, the convolution by $K$ extends to a bounded operator on $L^2(\mathbb{R}^{3})$ and coincide in Fourier space with the multiplication by a bounded function $\widehat{K} \in L^\infty(\mathbb{R}^{3})$. By definition, for $f\in L^p(\mathbb{R}^{3})$, 
\begin{equation*}
K\star f := \lim_{\substack{\epsilon\to 0 \\ \eta\to\infty}} K_{\varepsilon}^\eta \star f \text{ where } K_{\varepsilon}^\eta(x) = K(x) \mathds{1}_{\varepsilon \leq |x|\leq \eta}(x), \,\,\,\forall x \in \mathbb{R}^{3}.
\end{equation*}
As computed in Lemma \ref{lem_dip_cont}
\begin{equation}\label{K_chapeau}
\widehat{K}^{\eta}_\varepsilon = \int_{\mathbb{S}^2} \int_\varepsilon^\eta \big(\cos(r p \cdot \omega) - 1 \big)\Omega(\omega) \frac{dr}{r} d\sigma(\omega).
\end{equation}
In order to prove (\ref{eq_lem_HTNLS}), we seek to find a bound on the error $\int(K_0^{RN^{-\beta}} \star |u|^2)|u|^2$. First, let us notice that for $p\in\mathbb{R}^{3}$
\begin{align}\label{remind_K}
\left| \widehat{K}_0^{RN^{-\beta}} (p) \right| &\leq C  \int_{\mathbb{S}^2} \int_0^{RN^{-\beta}}  \left| \cos\left(r p \cdot \omega \right) - 1 \right| \frac{dr}{r} d\sigma(\omega) \leq  C |p| N^{-\beta}.
\end{align}
Then,
\begin{align}\label{ineq_rest_K_Hartree}
	\bigg| \int_{\mathbb{R}^{3}} (K_N \mathds{1}_{|x| \geq R_0 N^{-\beta}} \star |u|^2) |u|^2 - \int_{\mathbb{R}^{3}} (K \star |u|^2) |u|^2 \bigg|
	&\leq  \int_{\mathbb{R}^{3}} \widehat{K}_0^{R_0 N^{-\beta}} \left|\widehat{|u|^2}\right|^2 \nonumber \\
	&\leq C N^{-\beta} \int_{\mathbb{R}^{3}} |p| \big|\widehat{|u|^2}(p)\big|^2.
\end{align}
By a density argument we can assume that $u$ is smooth and use that $\widehat{|u|^2} = \widehat{|u|}\star\widehat{|u|}$, we obtain
\begin{align*}
\int_{\mathbb{R}^{3}} |p| \big|\widehat{|u|^2}(p)\big|^2 &\leq \int |p|\widehat{|u|}(p-k)\widehat{|u|}(k)\widehat{|u|}(p-q)\widehat{|u|}(q) dqdkdp \\
&\leq 2 \int |p-k|\widehat{|u|}(p-k)\widehat{|u|}(k)\widehat{|u|}(p-q)\widehat{|u|}(q) dqdkdp,
\end{align*}
where we used that $|p|\leq |p-k| + |k|$ and the symmetry of the integrand. We now define $f = \mathcal{F}^{-1}\left(\widehat{|u|}\right)$ and $g = \mathcal{F}^{-1}\left(|p| \widehat{|u|}\right)$ such that $\|f\|_{H^1} = \|u\|_{H^1}$ and $\|g\|_{L^{2}} \leq \|\nabla u\|_{L^{2}}$. And finally,
\begin{align}\label{ineq_majo_moment_Hartree}
\int_{\mathbb{R}^{3}} |p| \big|\widehat{|u|^2}\big|^2 &\leq \int |g(x)||f(x)|^3 dx \nonumber \\
&\leq \|g\|_{L^{2}} \|f\|_{L^{6}} \leq \|g\|_{L^{2}} \|f\|_{H^1} \leq \|u\|_{H^1}^4
\end{align}
Inequality (\ref{ineq_majo_moment_Hartree}) inserted in (\ref{ineq_rest_K_Hartree}) concludes the proof of Lemma \ref{lem_from_H_to_nls}.
\end{proof}

\section*{Appendices}
\appendix
\section{Proof of Proposition \ref{prop_a_b}}\label{section_proof_prop_1}
Here, we show that in the dipolar case $K=K_{dip}$, any admissible $(a,b)$, i.e satisfying the condition (\ref{hyp1_theo_GP}) of Theorem \ref{theo_GP}, comes from a pair potential in the many-body problem. The proof uses two well chosen potentials $w_{dip}$ and $\widetilde{w_{dip}}$ which give the case of equality in (\ref{hyp1_theo_GP}) respectively for $b>0$ and $b<0$. Higher $a$'s are then achieved by adding a non negative pair function having the correct mass. First, we need the following lemma in order to build a classically stable potential behaving like the dipolar one at large distances.
\begin{lem}[\cite{Lewin-04b}, Dipolar approximation]\label{lem_dip_approx}
There is a constant $C>0$ such that for all $x, h \in \mathbb{R}^{3}$ with $R + h \neq 0$,
\begin{equation*}
\bigg| \frac{1}{|x + h|} - \bigg( \frac{1}{|x|} - \frac{e_x\cdot h}{|x|^2} 
+ \frac{3(e_x \cdot h)^2 - |h|^2)}{2|x|^3} \bigg)\bigg| \leq \frac{C |h|^3}{|x|^3 |x+h|}
\end{equation*}
with $e_x = x/|x|$.
\end{lem}
We now recall the definition of $w_{dip}$,
\begin{equation}\label{def_w_dip}
w_{dip}(x) = 2 W(x) - W(x+dn) - W(x-dn) \text{ with } W(x) = \frac{1-e^{-|x|}}{|x|}.
\end{equation}
The potential $w_{dip}$ represents the total interaction of two dipoles given by four charged particles interacting via the smeared Coulomb potential $W$ whose position and charge are $(\mathcal{O},1), (dn,-1)$ (for the first dipole) and $(x,1),(x+dn,-1)$ (for the second dipole). The dipolar moment $d$ is a characteristic of the system and is parallel to $n$, the common direction of the dipoles. In particular, by Lemma \ref{lem_dip_approx}, if we fix some $R>0$, the potential $w_{dip}$ can be written 
\begin{equation}\label{def_w_2}
w_{dip} = w_0 + d^2 \mathds{1}_{|x|\geq R} K
\end{equation}
with $w_0\in L^1\cap L^\infty$. Furthermore, $w_{dip}$ is of positive type since
\begin{align}\label{w_dip_pos_type}
\widehat{w_{dip}}(k) &= 2 \widehat{W}(k) -\widehat{W(\cdot + dn)}(k) - \widehat{W(\cdot -dn)}(k) \nonumber \\
&= \frac{8 \pi}{|k|^2(1+ |k|^2)} \left( 1 - \cos(k\cdot dn) \right) \geq 0.
\end{align}

\subsubsection*{Computation of the short range parameter $a$}

Now, let us compute the short range strength $a = \int w_0$, defined in (\ref{def_w_2}). Recall that this quantity is independent of the choice of $R>0$. We have
\begin{equation*}
\int_{\mathbb{R}^{3}} w_0 = \int_{\mathbb{R}^{3}} \left(w_{dip} - d^2\mathds{1}_{|x|>R} K\right) = \lim_{A\to\infty} \int_{B(0,A)} w_{dip},
\end{equation*}
thanks to the cancellation property (\ref{cancellation_prop_0}).
\begin{lem}[Computation of the short range strength]\label{lem_comput_sr}
The short range parameter of $w_{dip}$ is given by
\begin{equation*}
\int_{\mathbb{R}^{3}} w_0 = \frac{4\pi}{3} d^2.
\end{equation*}
\end{lem}
\begin{proof}
Recall that $w - d^2\mathds{1}_{|x|>R} K \in L^1$, we can consider
\begin{equation}\label{compute_sr}
\mathcal{F}\left(w_{dip} - d^2\mathds{1}_{|x|>R} K\right) = \widehat{w_{dip}} - d^2\left(\widehat{K} + d^2 \widehat{K} - \widehat{K\mathds{1}_{|x|>R}}\right)
\end{equation}
which is a continuous function and satisfies $\mathcal{F}\left(w_{dip} - d^2\mathds{1}_{|x|>R} K\right)(0) = \int w_0$. Now, we compute the limit when $p\to0$ of the right side of  (\ref{compute_sr}). The third term is an error term and according to (\ref{remind_K}), it is a $\mathcal{O}(p)$ as $p\to 0$. On the other hand, for any $p\in\mathbb{R}^{3}$, we can write the first term as
\begin{align*}
\mathcal{F}(w_{dip})(p) - d^2\mathcal{F}(K)(p) =\, &4\pi \left[ \frac{2}{(1+|p|^2)|p|^2}(1-\cos(p\cdot dn)) - \frac{(p\cdot dn)^2}{|p|^2} + \frac{d^2}{3} \right],
\end{align*}
which tends to $\frac{4\pi}{3}d^2$ as $p\to 0$.
\end{proof}

Now, using $w_{dip}$, we will show that any $a$ and $b$ satisfying assumptions (\ref{hyp1_theo_GP}) are reached for some potential $w$. For this purpose, let us distinguish two cases.
\subsubsection*{Case 1: $b > 0$}
Let $(a,b)$ satisfy condition (\ref{hyp1_theo_GP}) with $b>0$. Let us take $d^2 = b$ and $a\geq 4\pi b / 3$. Define
$$
w = f + w_{dip},
$$
with $f\in L^1(\mathbb{R}^{3})\cap L^{2}(\mathbb{R}^{3})$ of positive type, such that $\int f = a - 4\pi b / 3$. Then, the potential $w$ is classically stable, since $w \geq w_{dip}$ and $w_{dip}$ is. It satisfies the assumptions of Theorem \ref{theo_1}, and so the result holds for $b>0$.
\subsubsection*{Case 2: $b < 0$}
Without loss of generality, let us assume $n = e_3$ and let us take $d^2 = - b$. Define 
$$\widetilde{w_{dip}} = f-w_{dip}$$
with $f$ being the inverse Fourier transform of
\begin{equation*}
\widehat{f}(p) = \frac{4\pi}{\left(1+p^2\right)\left(d^{-2} + \frac{p_1^2 +p_2^2}{4}\right)}.
\end{equation*}
We have $\widehat{f}\in L^1$, since
\begin{equation*}
\int_{\mathbb{R}^3} \frac{1}{\left(1+p^2\right)\left(1+p_1^2+p_2^2\right)} dp_1dp_2dp_3
= \int_{\mathbb{R}} \frac{dp_3}{1+p_3^2} \int_{\mathbb{R}^2} \frac{dp_1 dp_2}{\left(1+p_1^2+p_2^2\right)^{3/2}} < \infty.
\end{equation*}
One can verify that $(1-\Delta)^2 \widehat{f}\in L^1$ and therefore $f \in L^1$. Furthermore, the inequality
\begin{equation*}
2\left(1-\cos(p_3 d)\right)\left(d^{-2}+\frac{p_1^2+p_2^2}{4}\right) \leq p^2,
\end{equation*}
implies $\widehat{\widetilde{w_{dip}}} = \widehat{f}- \widehat{w_{dip}} \geq 0$. Moreover, since $\widehat{f} \in L^{1}$, $f\in L^1\cap L^2$ the potential $\widetilde{w_{dip}}$ satisfies the assumptions of Theorem \ref{theo_1} and its short range parameter is given by $a= \widehat{(f-w_0)} (0) = -8\pi b/3$. 

This concludes the case of equality of Proposition \ref{prop_a_b}, the other cases are obtained by adding positive functions to the potential $\widetilde{w_{dip}}$.
 
\section{Stabilization of the long range potential}\label{section_proof_lem_stab_dip}

Here we show that given a potential $K=\Omega(x/|x|)|x|^{-3}$, with $\Omega$ an even function satisfying the cancellation property (\ref{cancellation_prop_0}) and some regularity condition, there always exists a classically stable potential $w_{stab}$ behaving like $K$ at infinity (Proposition \ref{lem_existence_dip}). Moreover, if some potential $w$ has the long range behavior of $K$, we show that one can always make it classically stable by increasing its value in a fixed neighborhood of the origin (Proposition \ref{lem_stab_dip}).

\begin{prop}[Existence of stable potential]\label{lem_existence_dip}
Let $K = \Omega(x/|x|)|x|^{-3}$ with $\Omega\in W^{4,q}(\mathbb{S}^{2})$, an even function, satisfying the cancellation property (\ref{cancellation_prop_0}), for some $q>1$. Then, for any $R>0$, there exists some continuous function $w_{stab}\in L^p(\mathbb{R}^{3})$, for $1<p\leq \infty$, of positive type (hence classically stable), such that $w_{stab} - \mathds{1}_{|x|>R} K \in L^1(\mathbb{R}^{3})$. Moreover, we can choose $w_{stab}$ such that $|w_{stab}(x) - \mathds{1}_{|x|>R}K(x)| = \mathcal{O}(e^{-\mu |x|})$, for some $\mu>0$, as $x$ tends to infinity.
\end{prop}
\begin{prop}[Long range stability]\label{lem_stab_dip}
Let $K$ be as in Proposition \ref{lem_existence_dip}, $R>0$ and $w\in L^1(\mathbb{R}^{3})$ with $(w)_- \in L^\infty(\mathbb{R}^{3})$ and such that $w(x) - \mathds{1}_{|x|>R} K(x) \in L^1(\mathbb{R}^{3})$. For all $\delta >0$ there exists $\eta > 0$ such that  $w + \eta \mathds{1}_{|x|< \delta}$ is classically stable.
\end{prop} 
To show the properties above we need the following lemma whose proof will be postponed.
\begin{lem}[Stabilisation by a short range potential] \label{lem_stab_longue_portee}
Let $w_0 \in L^{1}_{\loc}(\mathbb{R}^{3})$, $w_0 > \eta > 0$ and $w_1 \in L^\infty(\mathbb{R}^{3})$ such that 
\begin{equation}\label{hyp_stab_w_1}
|w_1(x)| \leq C_0 1/|x|^{3+\epsilon}
\end{equation}
for all $x \in \mathbb{R}^{3}$, where $\eta, \epsilon, C_0 > 0$.
Let us define for $R > 0$
$$
	\phi_{R} = w_0 \mathds{1}_{|x| \leq R} + w_1 \mathds{1}_{|x|\geq R}.
$$
Then, there is come constant $C_{\varepsilon}$, given by (\ref{Cuniv}), such that, if 
\begin{equation}\label{cond_stab}
C_{\varepsilon} \leq \frac{R^{2+\epsilon} \eta}{C_0(1+R^{-6\left(1+\tfrac{2}{\varepsilon}\right)})},
\end{equation}
then, there exists $\psi \in L^1(\mathbb{R}^{3})$ such that $\phi_{R_0}\geq \psi$, $\widehat{\psi} \geq 0$ and $\widehat{\psi} \in L^1$. In particular, $\phi_{R_0}$ is classically stable.
\end{lem}
If Lemma \ref{lem_stab_longue_portee} and Proposition \ref{lem_existence_dip} hold, then Proposition \ref{lem_stab_dip} is easily verified.
\begin{proof}[Proof of Proposition \ref{lem_stab_dip}]
Let us consider $w_{stab}$ as in Proposition \ref{lem_existence_dip}.  Let us write $w = w_0\mathds{1}_{|x| \leq R} + \mathds{1}_{|x|>R} K$ and rewrite $w$ as,
\begin{equation}\label{w_decompo_stab}
w = w_0\mathds{1}_{|x| \leq R} - w_{stab}\mathds{1}_{|x|<R} + (K-w_{stab}) \mathds{1}_{|x|>R} + w_{stab}.
\end{equation}
The second term in (\ref{w_decompo_stab}) is smaller than $\|w_{stab}\|_{L^\infty}$, the third term is a $\mathcal{O}(e^{-|x|})$ and the last term is of positive type according to Proposition \ref{lem_existence_dip}. Then, by increasing $w_0$ sufficiently in some fixed neighborhood of the origin, one can make $w$ classically stable by using Lemma \ref{lem_stab_longue_portee}. The idea of the proof is based on a result of Ruelle \cite[Proposition 3.2.8]{Ruelle}.
\end{proof}
Now we prove Lemma \ref{lem_stab_longue_portee} and  Proposition \ref{lem_existence_dip}.
\begin{proof}[Proof of Lemma \ref{lem_stab_longue_portee}.]
The strategy of the proof is to increase the function $\phi_R$, by adding a function $\psi_1$ whose Fourier transform is well controlled, and such that $\phi_R + \psi_1$ is positive. We will then look for conditions on $R,\eta$ and $C_0$ to bound $\phi_R + \psi_1$ from below by a function of positive type $\psi_2$, such that $\psi:=\psi_2 - \psi_1$ is classically stable.\\
Let $\alpha$ be a smooth positive function with support in $B(0,1)$ and such that $\int \alpha = 1$. Let us denote $\alpha_R(x) = (R/2)^{-3} \alpha(2x/R).$ For $|x|>R$, we have
\begin{align}\label{ineq_stab_gen}
\int_{B(0,R/2)} \frac{1}{|x-x'|^{3+\epsilon}} \alpha_R(x') dx' &\geq \frac{1}{\left(|x|+R/2\right)^{3+\epsilon}} \geq \frac{R^{3+\epsilon}}{\left(R+R/2\right)^{3+\epsilon}} \frac{1}{|x|^{3+\epsilon}} \nonumber \\ 
&\geq \left(\frac{2}{3}\right)^{3+\varepsilon} \frac{1}{|x|^{3+\varepsilon}}.
\end{align}
Let us define 
$$\psi_1(x) = C_0\left(\frac{3}{2}\right)^{3+\varepsilon}  \left(\frac{1}{|x'|^{3+\epsilon}} \mathds{1}_{|x'|\geq R/2}\right) \star \alpha_R(x),$$
 from (\ref{hyp_stab_w_1}) and (\ref{ineq_stab_gen}) we deduce $$|w_1(x)| \mathds{1}_{|x|>R} \leq C_0\frac{1}{|x|^{3+\epsilon}} \mathds{1}_{|x|>R} \leq \psi_1(x)$$ and in particular we obtain
\begin{equation}\label{psi_1}
 \psi_1 \leq w_1\mathds{1}_{|x|>R} + 2\psi_1.
\end{equation}
We now look for a bound on $\widehat{\psi_1}$, we have
\begin{equation*}
\|(1+p^2)^3 \widehat{\psi_1}\|_{\infty} \leq \|(1-\Delta)^3 \psi_1\|_{L^{1}}
\end{equation*}
from which we obtain
\begin{equation}\label{ineq_stab_psi_1_four}
| \widehat{\psi_1}(p) | \leq \frac{C_0 3^{3+\epsilon} \|(1-\Delta)^3 \alpha_R \|_{L^1} \|\mathds{1}_{|x|\geq R} 1/|x|^{3+\epsilon}\|_{L^1} }{2^{3+\epsilon} (1+p^2)^3} = \frac{4\pi C_0 3^{3+\epsilon} \|(1-\Delta)^3 \alpha_R \|_{L^1}}{2^{3+\epsilon} R^{\varepsilon}(1+p^2)^3}.
\end{equation}
Let us define for $\mu>0,\;\lambda > 1$, 
$$\psi^{\mu}_\lambda (x) = \mu \frac{e^{- \lambda|x|/2} - e^{- \lambda|x|}}{|x|}.$$
We have
\begin{equation}\label{ineq_stab_psi_mu_four}
\widehat{\psi^{\mu}_\lambda} (p) = 
\mu\left( \frac{1}{|p|^2 + \lambda^2/4} - \frac{1}{|p|^2 + \lambda^2} \right) =
\frac{\frac{3}{4}\lambda^2\mu}{(|p|^2+\lambda^2/4)(|p|^2 + \lambda^2)}
\geq \frac{\frac{3}{4}\lambda^2\mu}{(|p|^2+\lambda^2)^2}.
\end{equation}
We want to find a condition on $\lambda$ and $\mu$ such that $\psi^\mu_\lambda(x) \leq \psi_1(x)$ for $|x|>R$, and since it is sufficient to have
\begin{equation*}
\psi^{\mu}_\lambda(x) = \mu \frac{e^{- \lambda|x|/2} - e^{- \lambda|x|}}{|x|} \leq \mu \frac{e^{-\lambda|x|/2}}{|x|} \leq C_0 \frac{1}{|x|^{3+\epsilon}} \leq \psi_1(x) \quad \forall \; |x| > R,
\end{equation*}
it is then sufficient to assume $\mu \leq C_0 C_{\varepsilon} \lambda^{2+\varepsilon}$ where $C_{\varepsilon} = (4+2\varepsilon)^{2+\varepsilon}e^{-(2+\varepsilon)}$. Therefore, with (\ref{psi_1}) we obtain
\begin{equation} \label{stab_maj}
\psi^{\mu}_\lambda(x)-2\psi_1(x) \leq w_1(x), \quad \forall\; |x|>R.
\end{equation}
On the other hand for $\mu \leq \eta$, we have
\begin{equation*}
\psi^\mu_\lambda(x) - 2\psi_1(x) \leq \mu \leq \eta \leq w_0(x), \quad \forall \;|x|<R.
\end{equation*}
From (\ref{ineq_stab_psi_1_four}) and (\ref{ineq_stab_psi_mu_four}) we deduce that if $\mu R^{\epsilon} \geq 3^{4} \pi  \lambda^2  C_0 \|(1-\Delta)^3 \alpha_R \|_{L^1}$ and $\lambda > 1$, then
\begin{equation*}
0 \leq \widehat{\psi^{\mu}_\lambda} - 2 \widehat{\psi_1}.
\end{equation*}
Gathering the conditions, we deduce that if
\begin{equation}\label{stab_condition_0}
3^4 \pi \frac{\|(1-\Delta)^3 \alpha_R \|_{L^1}^{1 + \tfrac{2}{\varepsilon}}}{(4+2\varepsilon)^{\tfrac{4+2\varepsilon}{\epsilon}} e^{-\tfrac{4+2\varepsilon}{\epsilon}}} < \frac{R^{2+\varepsilon} \eta}{C_0}
\end{equation}
then we can find some $\lambda,\mu> 0$ such that $\psi := \psi^{\mu}_\lambda -2\psi_1 \leq \phi_{R}$ where $\psi$ fulfills the assumptions of the lemma. We define
\begin{equation}\label{Cuniv}
C_{\varepsilon,0}= \frac{3^4 \pi e^{\tfrac{4+2\varepsilon}{\epsilon}}}{(4+2\varepsilon)^{\tfrac{4+2\varepsilon}{\epsilon}}} 
\sup_{\substack{\\ R>0}} \inf_{\substack{\\ \alpha}} \frac{\|(1-\Delta)^3 \alpha_{R} \|_{L^1}^{1 + \tfrac{2}{\varepsilon}}}{1+R^{-6\left(1+\tfrac{2}{\varepsilon}\right)}} .
\end{equation}
\end{proof}

\begin{proof}[Proof of Proposition \ref{lem_existence_dip}.]
Let $\chi$ be a smooth function such that $\chi(x) = 1$ for $|x|\leq 1/2$ and $\chi(x) = 0$ for $|x|\geq 1$. Let us denote by $\chi_R (x) = \chi(x/R)$ and $g = (1-\chi_R)K$. One can verify that $\partial^\alpha \Omega$ still satisfies the cancellation property (\ref{cancellation_prop_0}) for any multi-index $|\alpha| \leq 4$. One can also notice that $\partial^\alpha g \in L^q(\mathbb{R}^{3})$ is a sum of terms in $L^1(\mathbb{R}^{3})$ and terms satisfying the assumptions of Lemma \ref{lem_dip_cont} and hence of bounded Fourier transform. Therefore $\widehat{(1-\Delta)^2 g} = (1+|p|^2)^2 \widehat{g}$ and
\begin{equation*}
(1+|p|^2)^2 |\widehat{g}(p)| \leq \|(1-\Delta)^2 g\|_{L^1} \leq C_R \|\Omega\|_{W^{4,q}}.
\end{equation*}
Now, for $\mu >0$, consider $\psi^\mu_1 \in L^1(\mathbb{R}^{3})\cap L^\infty(\mathbb{R}^{3})$ as defined in the proof of Lemma \ref{lem_stab_longue_portee}. It satisfies
\begin{equation*}
\widehat{\psi^{\mu}_1} (p)
\geq \frac{\frac{3}{4}\mu}{(1+ |p|^2)^2}.
\end{equation*}
Taking $\mu \geq 4 C_R\|\Omega\|_q /3$, we define $w_{stab} = \psi^\mu_1 + (1-\chi_R)K$. Thanks to Sobolev embeddings $W^{4,q}(\mathbb{S}^{2}) \subset C^{1}(\mathbb{R}^{3})$ and we have $w_{stab} - \mathds{1}_{|x|>R}K \in L^\infty(\mathbb{R}^{3})$. Moreover $\psi_1^\mu (x) = \mathcal{O}(e^{-\mu|x|})$ as $x \to\infty$, this concludes the proof of Proposition \ref{lem_existence_dip}.
\end{proof}

\section{Proof of Theorem \ref{theo_uniq}: uniqueness} \label{section_proof_lem_uniq_0}
\subsubsection*{Proof of \textit{1})}

We use the implicit function theorem to prove the uniqueness of the ground state of the Gross-Pitaevskii energy in some neighborhood of $\varphi = 0$. This method does not allow to give an explicit value of $\varphi_0$ in the lemma. 

Let us fix some admissible parameters $(a,b)$, i.e. satisfying condition (\ref{hyp1_theo_GP}), and let us define 
$$\mathcal{E}(u,\varphi) := \int_{\mathbb{R}^{3}} |(\nabla + i \varphi A)u|^2 + \int_{\mathbb{R}^{3}} V |u|^2 + a \int_{\mathbb{R}^{3}} |u|^4 + b\int_{\mathbb{R}^{3}\times\mathbb{R}^{3}} (K \star |u|^2) |u|^2.$$
We assume without loss of generality that $V\geq 0$. A minimizer of $\mathcal{E}(\cdot,\varphi)$ satisfies the Euler-Lagrange equation
\begin{equation*}
\Phi(u_1,u_2,\varphi) := h^\varphi u + 2 a |u|^2 u + 2 b K\star |u|^2 u - \mu(u,\varphi) u = 0
\end{equation*}
where $u=u_1+iu_2$ with $u_1,u_2$ are real valued, $$\boxed{h^\varphi = (i\nabla + \varphi A)^2 + V}$$ and $\mu(u,\varphi) = \mathcal{E}(u,\varphi) + a \int |u|^4 + b \int K\star |u|^2 |u|^2$. We want to apply the implicit function on $\Phi$ in a neighborhood of $(u_0,0)$. For that, we first need the following lemma.

\begin{lem}\label{lem_H_0_uniq}
Define the operator
\begin{equation*}
L^- = -\Delta + V + 2 a u_0^2 + 2 b \, K \star u_0^2 - \mu(u_0,0)
\end{equation*}
with domain $\mathcal{D}(h_0)$. Then $L^-$ is a non-negative self-adjoint operator with compact resolvent. Moreover, its ground state is $u_0$ and it is non-degenerate.
\end{lem}
\begin{proof}
Thanks to Lemma \ref{lem_dip_cont} and Sobolev embeddings we have $f := 2a u_0^2 + 2 b (K\star u_0^2) \in L^2(\mathbb{R}^{3})$. 
The self-adjointness and the compactness of the resolvent of $H_0$ follow from Sobolev inequality. The uniqueness and the non-negativity of the ground state up to a constant phase follow from \cite[Theorem 7.8]{LieLos-01}. Now, since $u_0$ minimizes $\mathcal{E}(\cdot,0)$, we know that $H_0 u_0 = \mu_0 u_0$ and also that for any $v\in\mathcal{D}(h_0^{1/2})$ we have $\frac{d^2}{dt^2}\mathcal{E}(u_0+tv)|_{t=0}\geq0$. The latter is equivalent to
\begin{gather*}
\Big< v \Big| -\Delta + V+ 2 a u_0^2 + 2 b K\star u_0^2 \Big|  v \Big>  \\
+ 4 a \int u_0(x)^2 \Re(v(x))^2dx
+ 4 b \iint K(x-y) u_0(x) \Re(v(x)) u_0(y) \Re(v(y))dxdy \\ \geq \mu(u_0,0) \|v\|_{L^2(\mathbb{R}^{3})}^2.
\end{gather*}
Taking $i v$ with $v$ real in the last inequality shows that $L^- \geq 0$ (since we know the ground state of $L^-$ is real up to a constant phase). Taking $v = u_0$ shows that $u_0$ is the ground state.
\end{proof}
To verify the assumptions of the implicit function theorem we must compute
\begin{multline*}
d_1\Phi _{(u_0,0,0)} \xi := L^+\xi = L^-\xi + 4 a u_0^2 \xi + 4b (K\star u_0\xi) u_0  - 4 \left( a \int u_0^3 \xi + b \int (K\star u_0^2)u_0\xi \right)u_0 \\ - 2 \mu(u_0,0) \braket{u_0,\xi}_{L^2} u_0
\end{multline*}
and
\begin{equation*}
d_2\Phi _{(u_0,0,0)} \xi = L^-\xi= h^0 \xi + 2 a u_0^2 \xi + 2 b (K\star u_0^2) \xi - \mu(u_0,0)\xi.
\end{equation*}
Define $\mathcal{V} := \mathcal{D}(h_0^{1/2})$ endowed with the norm $\|u\|_\mathcal{V} = \|u\|_{L^2} + \|h_0^{1/2}u\|_{L^2}$ and denote by $\{u_0\}^\perp$ the orthogonal space in $L^2(\mathbb{R}^{3})$ to $u_0$ with respect to the $L^2$ scalar product. We emphasize that the latter are spaces composed of real valued functions. Using Sobolev inequality and Lemma \ref{lem_dip_cont}, it is elementary that $\Phi$, $d_1\Phi $ and $d_2\Phi$ are continuous functions. In order to apply the implicit function theorem, we want to prove that $(L^+)^{-1} : (L^2(\mathbb{R}^{3}),\|\cdot\|_{L^2}) \to (\mathcal{V},\|\cdot\|_{\mathcal{V}})$ and $(L^-)^{-1} :  (\{u_0\}^\perp,\|\cdot\|_{L^2}) \to (\mathcal{V}\cap \{u_0\}^\perp,\|\cdot\|_{\mathcal{V}})$ are bounded linear operators. Since, $L^\pm$ have compact resolvants, it is sufficient to prove that $0$ is not en eigenvalue of $L^+$ and that $\ker(L^-) = \spn\{u_0\}$. We prove this by contradiction, let $\xi \in \ker(L^+)$,
writing $\xi = \lambda u_0 + \xi^\perp$ with $\xi ^\perp\in\{u_0\}^\perp$, we have
\begin{equation}\label{diff_0_1}
\begin{gathered}
\braket{\xi^\perp, L^+ \xi } = 0 = \braket{\xi^\perp, L^-\xi^\perp} + 4\left(a \int (u_0\xi^\perp)^2 + b \int K \star (u_0\xi^\perp) u_0\xi^\perp\right) - 2 \lambda \braket{u_0,h^0 \xi^\perp},
\end{gathered}
\end{equation}
\begin{gather}
\braket{u_0, L^+\xi } = 0 = - 2 \lambda \mu(u_0,0). \label{diff_0_2}
\end{gather}
From (\ref{diff_0_2}), and since $\mu(u_0,0)>0$ because we assumed $V\geq 0$, we deduce that $\lambda=0$ and then from (\ref{diff_0_1}) together with Lemma \ref{lem_H_0_uniq} and the admissibility of $(a,b)$ we obtain that $\xi^\perp = 0$ which is the contradiction we seek. Similar arguments gives that $\ker(L^-) = \spn (u_0)$. We in fact only proved that $(L^+)^{-1} : (L^2(\mathbb{R}^{3}),\|\cdot\|_{L^2})\to (\mathcal{V},\|\cdot\|_{L^2})$ is continuous, but since $\|\cdot\|_{L^2} \lesssim \|\cdot\|_{\mathcal{V}}$ the desired continuity property holds. Now, we can use \cite[Theorem 1.2.1 \& 1.2.3]{Chang} and claim that there is some neighborhood $X\times ]-\varphi_0,\varphi_0[ \ni (u_0,0)$ in $\mathcal{H} := \mathcal{V} \oplus \left(\{u_0\}^\perp\cap\mathcal{V}\right)\times \mathbb{R}$ and some function $h : ]-\varphi_0,\varphi_0[ \to X$ such that for any $(u,\varphi) \in X\times ]-\varphi_0,\varphi_0[ $
\begin{equation*}
\Phi(u,\varphi) = 0 \iff u = h(\varphi).
\end{equation*}
Hence, there is uniqueness of the minimizer in $X$. Now we must prove that there is some other neighborhood of $0$ which we will still denote by $]-\varphi_0,\varphi_0[$ such that for all $|\varphi|<\varphi_0$, any minimizer of $\mathcal{E}(\cdot,\varphi)$ (restricted to the unit sphere of $L^2(\mathbb{R}^{3})$) belongs to $X$ (up to a constant phase). We prove it by contradiction, assume there is a sequence $\varphi_n\to 0$ with $u_n\in \left(\mathcal{V}\oplus\mathcal{V}\right) \setminus X$ minimizing $\mathcal{E}(\cdot,\varphi_n)$ for all $n$. First, notice that for all $u\in L^2(\mathbb{R}^{3})$ there is some $\theta \in\mathbb{R}$ such that 
$$\Im(u e^{i\theta}) = \Re(u) \sin(\theta) + \Im(u) \cos(\theta) \in \{u_0\}^\perp.$$ 
It suffices to take $\theta = - \tan^{-1}(\braket{\Im(u),u_0}/ \braket{\Re(u),u_0})$. Hence, we can assume that $u_n \in \mathcal{H}$ for all $n$. Then, notice that $(u_n)$ is bounded in $\mathcal{H}$ it is thus precompact in $L^2(\mathbb{R}^{3})$ and we can assume without loss of generality that it converges to some function $u_\infty\in L^2(\mathbb{R}^{3})$. Besides, for any $u$ and $|\varphi|\leq C$ we have
\begin{equation*}
|\mathcal{E}(u,\varphi)-\mathcal{E}(u,0)| \leq |\varphi| |\braket{u, (\nabla A + A\nabla) u} + |\varphi|^2 \braket{u, A^2 u} \leq C |\varphi| \braket{u, (h_0+1) u}.
\end{equation*}
So that
\begin{equation*}
\mathcal{E}(u_n,0) = \mathcal{E}(u_n,\varphi_n) + o(1) \leq \mathcal{E}(u_0,\varphi_n) + o(1) = \mathcal{E}(u_0,0) + o(1).
\end{equation*}
And since $\mathcal{E}(\cdot,0)$ is lower semi-continuous, by passing to the limit above we obtain $\mathcal{E}(u_\infty,0) \leq \mathcal{E}(u_0,0)$ hence $u_\infty = u_0$ by uniqueness. We recall that $X$ is a $\mathcal{H}$-neighborhood of $u_0$ so we must prove the convergence of $u_n$ in the topology of $\mathcal{H}$. We know that $(h_0^{1/2}u_n)$ is bounded in $L^2(\mathbb{R}^{3})$ thus, without loss of generality, we assume it converges weakly to some limit we denote by $v_\infty$. Then, for any $\xi\in\mathcal{V}\oplus \mathcal{V}$ we have
\begin{equation*}
\braket{h_0^{1/2}u_n, \xi} = \braket{u_n, h_0^{1/2} \xi} \to \braket{u_0, h_0^{1/2} \xi} = \braket{h_0^{1/2} u_0, \xi}
\end{equation*}
proving that $v_\infty = h_0^{1/2}u_0$ since $\mathcal{V}\oplus \mathcal{V}$ is a dense subset of $L^2(\mathbb{R}^{3})$. Now, let us prove that the convergence is strong. For that, it suffices to prove the conservation of mass in the limit: $\|h_0^{1/2}u_n\| \to \|h_0^{1/2}u_0\|$ as $n\to\infty$. We have, by the Euler-Lagrange equation,
\begin{equation*}
\left| \|h_0^{1/2}u_n\| - \|h_0^{1/2}u_0\| \right| \leq \left| \mathcal{E}(u_n,0) - \mathcal{E}(u_0,0) \right| + \int \left||u_n|^4 - |u_0|^{4}\right| + \int \left| K\star u_n^2 u_n^2 - K\star u_0^2 u_0^2\right| \to 0,
\end{equation*}
as  $n\to \infty$. We already know that the first term tends to $0$, for the two other, it comes from Sobolev embeddings and Lemma \ref{lem_dip_cont}. Hence $u_n \to u_0$ in $\mathcal{H}$ and this is the contradiction we seek : there is some $\varphi_0>0$ such that if $u$ is a minimizer of $\mathcal{E}(\cdot,\varphi)$ then $u\in X$. This ends the proof.

\subsubsection*{Proof of \textit{2})}
This result is stated for the case $K=0$ in a remark of \cite[Ch. 7]{LieSeiSolYng-05}. Here, we give a proof working for any $K$ satisfying the assumption of Lemma (\ref{lem_dip_cont}) using the fact that the interaction term is convex in $\rho = |u|^2$. We write the Fourier decomposition of $u$ in cylindrical coordinates
$$
u(r,\theta,z) = \sum_{n\in\mathbb{Z}} c_n(r,z) e^{in\theta}, \quad \forall r,z>0, \: \theta \in [0,2\pi[,
$$
where
$$
c_n(r,z) = \frac{1}{2\pi}\int_0^{2\pi} u(r,\theta,z) e^{-in\theta} d\theta, \quad \forall n \in \mathbb{Z}.
$$
Then, we have
\begin{equation}\label{ineq_magn_unicite}
\begin{gathered}
\int_{\mathbb{R}^{3}} \left| (\frac{i}{r}\partial_{\theta} + a) u \right|^2 = \int_{-\infty}^\infty\int_{0}^\infty \int_0^{2\pi} \sum_{n\in\mathbb{Z}} |c_n(r,z)|^2 \left(\frac{n}{r} + a\right)^2 r drdz \\
\geq \int_{-\infty}^\infty\int_{0}^\infty \int_0^{2\pi} \sum_{n\in\mathbb{Z}} |c_n(r,z)  a|^2 r drdz =
\int_{\mathbb{R}^{3}} | a u|^2 
\end{gathered}
\end{equation}
where we used that if $\| r a\|_{\infty} \leq 1/2$ then $|n / r + a| \geq |a|$ for any $n\in \mathbb{Z}$. From inequality (\ref{ineq_magn_unicite}), we know that taking $\braket{|u|^2}_{\theta}^{1/2}$, where $\braket{\cdot}_\theta$ stands for the $\theta$-average, can only lower the $\theta$-component of kinetic energy. We know that it also lowers the other components of the kinetic energy which are convex in $\rho = |u|^2$ since they are not affected by $A$ \cite[Theorem 7.8]{LieLos-01} and proves that minimizers are non negative up to a constant phase. Besides, the interaction term is strictly convex in $\rho$, indeed using that $a+b\widehat{K} \geq 0$ and denote by $<F(\rho)>_{\rho} \, = (F(\rho_1)+F(\rho_2))/2$ for some $\rho_1,\rho_2$ and any fonction $F$, we have
\begin{gather*}
\left<a \int \rho^2 + b \int (K\star \rho) \rho\right>_\rho = \int\left(a+b\widehat{K}\right)\left<|\widehat{\rho}|^2\right>_\rho \\
\geq \int\left(a+b\widehat{K}\right)|\braket{\,\widehat{\rho}\,}_\rho|^2 = \int\left(a+b\widehat{K}\right)|\widehat{\braket{\rho}}_\rho|^2 \\
 = a \int \braket{\rho}_\rho^2 + b \int K\star \braket{\rho}_\rho \braket{\rho}_\rho.
\end{gather*}
The strict convexity of the functional gives the uniqueness of the minimizer by standard means.

\end{document}